\documentclass[pdflatex,sn-mathphys-num, iicol]{sn-jnl}

\usepackage{graphicx}%
\usepackage{multirow}%
\usepackage{amsmath,amssymb,amsfonts}%
\usepackage{amsthm}%
\usepackage{mathrsfs}%
\usepackage[title]{appendix}%
\usepackage{xcolor}%
\usepackage{textcomp}%
\usepackage{manyfoot}%
\usepackage{booktabs}%
\usepackage{algorithm}%
\usepackage{algorithmicx}%
\usepackage[noend]{algpseudocode}%
\usepackage{listings}%
\usepackage{xspace}
\usepackage{microtype}
\usepackage{dblfloatfix}

\theoremstyle{thmstyleone}%
\newtheorem{theorem}{Theorem}
\newtheorem{lemma}[theorem]{Lemma}
\theoremstyle{thmstyletwo}%
\newtheorem{remark}{Remark}%

\newtheorem{corollary}[theorem]{Corollary}
\newtheorem{observation}[theorem]{Observation}

\theoremstyle{thmstylethree}%
\newtheorem{definition}{Definition}%

\newcommand{\PC}{\mathcal{P}\xspace}
\newcommand{\RC}{\mathcal{R}\xspace}

\newcommand{\TC}{\mathcal{T}\xspace}
\newcommand{\SPM}{\mathrm{SPM}\xspace}
\newcommand{\dist}{\mathit{dist}\xspace}

\newcommand{\assumption}[1]{\textsf{A#1}\xspace}

\begin{document}

\title[Fast Reconfiguration for Programmable Matter]{Fast Reconfiguration for Programmable Matter}

\author[1]{\fnm{Irina} \sur{Kostitsyna}}\email{i.kostitsyna@tue.nl}
\equalcont{These authors contributed equally to this work.}

\author[2]{\fnm{Tom} \sur{Peters}}\email{t.peters1@tue.nl}
\equalcont{These authors contributed equally to this work.}

\author[2]{\fnm{Bettina} \sur{Speckmann}}\email{b.speckmann@tue.nl}
\equalcont{These authors contributed equally to this work.}

\affil[1]{\orgname{KBR at NASA Ames Research Centre}}
\affil[2]{\orgname{TU Eindhoven}, \city{Eindhoven}, \country{The Netherlands}}

\abstract{The concept of programmable matter envisions a very large number of tiny and simple robot particles forming a smart material.
Even though the particles are restricted to local communication, local movement, and simple computation, their actions can nevertheless result in the global change of the material's physical properties and geometry.

A fundamental algorithmic task for programmable matter is to achieve global shape reconfiguration by specifying local behavior of the particles.
In this paper we describe a new approach for shape reconfiguration in the \emph{amoebot} model.
The amoebot model is a distributed model which significantly restricts memory, computing, and communication capacity of the individual particles.
Thus the challenge lies in coordinating the actions of particles to produce the desired behavior of the global system.

Our reconfiguration algorithm is the first algorithm that does not use a canonical intermediate configuration when transforming between arbitrary shapes.
We introduce new geometric primitives for amoebots and show how to reconfigure particle systems, using these primitives, in a linear number of activation rounds in the worst case.
In practice, our method exploits the geometry of the symmetric difference between input and output shape: it minimizes unnecessary disassembly and reassembly of the particle system when the symmetric difference between the initial and the target shapes is small.
Furthermore, our reconfiguration algorithm moves the particles over as many parallel shortest paths as the problem instance allows.}


\maketitle

\bmhead{Acknowledgments}
The work in Section~\ref{sec:holes} has been started at the 2nd Bertinoro Workshop on Distributed Geometric Algorithms. We thank Jonas Harbig, Tim Ophelders, and George Skretas for their input and discussions.

\section{Introduction}
Programmable matter is a smart material composed of a large quantity of robot particles capable of communicating locally, performing simple computation, and, based on the outcome of this computation, changing their physical properties. 
Particles can move through a programmable matter system by changing their geometry and attaching to (and detaching from) neighboring particles.
By instructing the particles to change their local adjacencies, we can program a particle system to reconfigure its global shape.
Shape assembly and reconfiguration of particle systems have attracted a lot of interest in the past decade and a variety of specific models have been proposed~\cite{Cheung2011,Gmyr2018,Patitz2014,Woods2013,Geary2014,Demaine2018b,Naz2016,Piranda2018}.
We focus on the \emph{amoebot} model~\cite{Derakhshandeh2014}, which we briefly introduce below.
Here, the particles are modeled as independent agents collaboratively working towards a common goal in a distributed fashion.
The model significantly restricts computing and communication capacity of the individual particles, and thus the challenge of programming a system lies in coordinating local actions of particles to produce a desired behavior of the global system.

A fundamental problem for programmable matter is shape reconfiguration.
To solve it we need to design an algorithm for each particle to execute, such that, as a result, the programmable matter system as a whole reconfigures into the desired target shape.
Existing solutions first build an intermediate \emph{canonical configuration} (usually, a line or a triangle) and then build the target shape from that intermediate configuration~\cite{Derakhshandeh2016,DiLuna2019}.
However, in many scenarios, such as shape repair, completely deconstructing a structure only to build a very similar one, is clearly not the most efficient strategy.

We propose the first approach for shape reconfiguration that does not use a canonical intermediate configuration when transforming between two arbitrary shapes. 
Our algorithm exploits the geometry of the symmetric difference between the input shape $I$ and the target shape $T$.
Specifically, we move the particles from $I\setminus T$ to $T \setminus I$ along shortest paths through the overlap $I\cap T$ over as many parallel shortest paths as the problem instance allows.
In the worst case our algorithm works as well as existing solutions. However, in practice, our approach is significantly more efficient when the symmetric difference between the initial and the target shape is small.

\subparagraph{Amoebot model}
Particles in the amoebot model occupy nodes of a plane triangular grid~$G$.
A particle can occupy either a single node or a pair of adjacent nodes: the particle is contracted and expanded, respectively.
The particles have constant memory space, and thus have limited computational power.
They are disoriented (no common notion of orientation) and there is no consensus on chirality (clockwise or counter-clockwise).
The particles have no ids and execute the same algorithm, i.e. they are identical.
They can reference their neighbors using six (for contracted) or ten (for expanded particles) \emph{port identifiers} that are ordered clockwise or counter-clockwise modulo six or ten (see Figure~\ref{fig:model} (top)).
Particles can communicate with direct neighbors by sending messages over the ports.
Refer to Daymude et al.~\cite{Daymude2021} for additional details.

\begin{figure}
	\centering
	\includegraphics{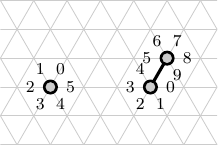}\hfil
	\includegraphics{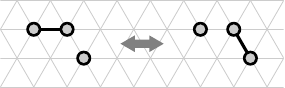}
	\caption{
		Top: particles with ports labeled, in contracted and expanded state.
		Bottom: handover operation between two particles.
	}
	\label{fig:model}
\end{figure}

Particles can move in two different ways: a contracted particle can \emph{expand} into an adjacent empty node of the grid, and an expanded particle can \emph{contract} into one of the nodes it currently occupies.
Each node of $G$ can be occupied by at most one particle, and we require that the particle system stays connected at all times.
To preserve connectivity more easily, we allow a \emph{handover} variant of both move types, a simultaneous expansion and contraction of two neighboring particles using the same node (Figure~\ref{fig:model} (bottom)).
A handover can be initiated by any of the two particles; if it is initiated by the expanded particle, we say it \emph{pulls} its contracted neighbor, otherwise we say that it \emph{pushes} its expanded neighbor.

Particles operate in activation cycles: when activated, they can read from the memory of their immediate neighbors, compute, send constant size messages to their neighbors, and perform a move operation.
Particles are activated by an asynchronous adversarial but fair scheduler (at any moment in time $t$, any particle must be activated at some time in the future $t'>t$).
If two particles are attempting conflicting actions (e.g., expanding into the same node), the conflict is resolved by the scheduler arbitrarily, and exactly one of these actions succeeds.
We perform running time analysis in terms of the number of \emph{rounds}: time intervals in which all particles have been activated at least once.

We say that a \emph{particle configuration}~$\PC$ is the set of all particles and their internal states.
Let $G_\PC$ be the subgraph of the triangular grid $G$ induced by the nodes occupied by particles in $\PC$.
Let a \emph{hole} in $\PC$ be any interior face of $G_\PC$ with more than three vertices.
We say a particle configuration~$\PC$ is \emph{connected} if there exists a path in $G_\PC$ between any two particles.
A particle configuration~$\PC$ is \emph{simply connected} if it is connected and has no holes.

\subparagraph{Related work}
The amoebot model, which is both natural and versatile, was introduced by Derakhshandeh et al.~\cite{Derakhshandeh2014} in 2014 and has recently gained popularity in the algorithmic community.
A number of algorithmic primitives, such as leader election~\cite{derakhshandeh2015leader,DiLuna2019,Dufoulon2021}, spanning forests~\cite{derakhshandeh2015leader}, and distributed counters~\cite{Daymude2020,Porter2018}, were developed to support algorithm design.

Derakhshandeh et al.~\cite{Derakhshandeh2016} designed a reconfiguration algorithm for an amoebot system which starts from particles forming a large triangle and targets a shape consisting of a constant number of unit triangles (such that their description fits into the memory of a single particle).
In their approach the initial large triangle is partitioned into the unit triangles, which move in a coordinated manner to their corresponding position within the target shape.
Derakhshandeh et al. make some assumptions on the model, including a sequential activation scheduler (at every moment in time only one particle can be active), access of particles to randomization, and common particle chirality.
Due to these assumptions, and the fact that the initial shape is compact, the reconfiguration process takes $O(\sqrt{n})$ number of rounds for a system with $n$ particles.

Di Luna et al.~\cite{DiLuna2019} were the first to reconfigure an input particle system into a line (a canonical intermediate shape). They then simulate a Turing machine on this line, and use the output of the computation to direct the construction of the target structure in $O(n\log n)$ rounds.
Their main goal was to lift some of the simplifying assumptions on the model of~\cite{Derakhshandeh2016}: their algorithm works under a synchronous scheduler, is deterministic, does not rely on the particles having common chirality, and  requires only the initial structure to be simply connected.
However, just as the work by Derakhshandeh et al.~\cite{Derakhshandeh2016}, their method only works for structures of constant description size.

Cannon et al.~\cite{Cannon2016} consider a stochastic variation of the amoebot model.
Viewed as an evolving Markov chain, the particles make probabilistic decisions based on the structure of their local neighborhoods.
In this variant, there exist solutions for compressing a system into a dense structure~\cite{Cannon2016}, simulating ant behavior to build a shortcut bridge between two locations~\cite{AndresArroyo2018}, and separating a system into multiple components by the \emph{color} of particles~\cite{Cannon2019}.

\subparagraph{Problem description}\label{sec:problem_description}
An instance of the \emph{reconfiguration problem} consists of a pair of connected shapes $(I, T)$ embedded in the grid $G$ (see Figure~\ref{fig:particle_reconfiguration_example}).
The goal is to transform the initial shape $I$ into the target shape $T$.
Initially, all particles in $I$ are contracted.
The problem is solved when there is a contracted particle occupying every node of $T$.

We make a few assumptions on the input and on the model.
Most of our assumptions fall into at least one of the following categories: they are natural for the problem statement, they can be lifted with extra care, or they are not more restrictive than existing work.

\begin{figure}[t]
    \centering
    \includegraphics{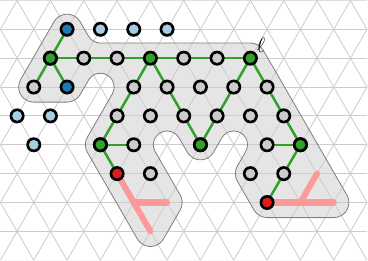}
    \caption{The particles form the initial shape $I$. The target shape $T$ is shaded in gray. Supply particles are blue, supply roots dark blue. Demand roots (red) store spanning trees of their demand components. The graph $G_L$ on the coarse grid is shown in green and leader~$\ell$ is marked. Particles on $G_L$ that are green are grid nodes, other particles on $G_L$ are edge nodes.}
    \label{fig:particle_reconfiguration_example}
\end{figure}

\subparagraph{Assumptions on $I$ and $T$}
We assume that the input shape $I$ and the target shape $T$ have the same number $n$ of nodes.
We call $I\cap T$ the \emph{core} of the system, the particles in $I\setminus T$ the \emph{supply}, and the particles in $T\setminus I$ the \emph{demand}.
In our algorithm, the core nodes are always occupied by particles, and the supply particles move through the core to the demand.

We assume that
\begin{description}
    \item[\assumption{1}] the core $I\cap T$ is a non-empty simply connected component. This is a natural assumption in the shape repair scenario when the symmetric difference between the initial and the target shape is small.
    \item[\assumption{2}] each demand component $D$ of $T\setminus I$ has a constant description complexity.
\end{description}

\noindent
In Section~\ref{sec:holes} we discuss a possible strategy for lifting \assumption{1}. Below \assumption{6} we explain why \assumption{2} is not more restrictive than the input assumptions in the state-of-the art.

\subparagraph{Assumptions on initial particle state}
A standard assumption is that initially all particles have the same state.
In this paper we assume that a preprocessing step has encoded $T$ into the particle system and thus the states of some particles have been modified from the initial state.
The reasons for this assumption are twofold: first, to limit the scope of this paper we chose to focus on the reconfiguration process; and second, the encoding preprocessing step, whose specification we omit, can be derived in a straightforward manner from existing primitives and algorithms. Below we describe more precisely what our assumptions are on the outcome of the preprocessing step.

To facilitate the navigation of particles through the core $I\cap T$, and in particular, to simplify the crossings of different flows of particles, we use a coarsened (by a factor of three) version of the triangular grid. 
Let $G_L$ be the intersection of the coarse grid with $I\cap T$ (the green nodes Figure~\ref{fig:particle_reconfiguration_example}).
We assume that after the preprocessing step
\begin{description}
    \item[\assumption{3}]\label{a3} a coarsened grid $G_L$ has been computed. By definition every node in the core is either in $G_L$ or is adjacent to a node in $G_L$.
    For simplicity of presentation we assume that the shape of the core is such that $G_L$ is connected;
    \item[\assumption{4}]\label{a4} particles know whether they belong to the supply $I\setminus T$, the demand $I\cap T$, or to the core $I\cap T$;
    \item[\assumption{5}]\label{a5} every connected component $C$ of the supply has a representative particle in the core $I \cap T$ adjacent to $C$, the \emph{supply root} of $C$; correspondingly, every connected component $D$ of the demand has one designated particle from the core adjacent to $D$, the \emph{demand root} of $D$;
    \item[\assumption{6}] each demand root $d$ stores a complete spanning tree rooted at $d$ of the corresponding demand component $D$ in its memory explicitly or implicitly, by encoding construction rules for $D$.
\end{description}

\noindent
If the core does not naturally induce a connected $G_L$, we can restore connectivity by locally deviating from the grid lines of the coarse grid and adapting the construction and use of $G_L$ accordingly.
Assumptions \assumption{2} and \assumption{4}--\assumption{6} allow us to focus our presentation on the reconfiguration algorithm itself.
Encoding a general target shape into the initial structure remains a challenging open problem. However, compared to prior work, our assumptions are not very strong: all other existing algorithms support only fixed, simple target shapes which can be easily encoded in the initial shape. For example, Derakhshandeh et al.~\cite{Derakhshandeh2016} assume that the initial shape $I$ is a large triangle and that the target shape $T$ consists of only a constant number of unit triangles. Here a leader particle can easily compute the information necessary for the pre-processing step and broadcast it through the system.

\subparagraph{Contribution and organization}
We present the first reconfiguration algorithm for programmable matter that does not use a canonical intermediate configuration. Instead, our algorithm introduces new geometric primitives for amoebots which allow us to route particles in parallel from supply to demand via the core of the particle system. 
A fundamental building block of our approach are \emph{feather trees}~\cite{disc22} which are a special type of \emph{shortest path trees} (SP-trees). 
SP-trees arrange particles in a tree structure such that the paths from leaves to the root are shortest paths through the particle structure.
The unique structure of feather trees allows us to use multiple overlapping trees in the particle system to enable particle navigation along shortest paths between multiple sources of supply and demand. In Section~\ref{sec:SPtrees} we give all necessary definitions and show how to efficiently construct SP-trees and feather trees using a grid variant of shortest path maps~\cite{Mitchell1991}.

In Section~\ref{sec:supply-demand} we then show how to use feather trees to construct a \emph{supply graph}, which directs the movement of particles from supply to demand. In Section~\ref{sec:nav} we describe in detail how to navigate the supply graph and also discuss the coarse grid $G_L$ which we use to ensure the proper crossing of different flows of particles travelling in the supply graph.
In Section~\ref{sec:algo} we summarize and analyze our complete algorithm for the reconfiguration problem.
We show that using a sequential scheduler we can solve the particle reconfiguration problem in $O(n)$ activation rounds. When using an asynchronous scheduler our algorithm takes $O(n)$ rounds in expectation, and $O(n \log n)$ rounds with high probability.
Finally, in Section~\ref{sec:holes}, we show how we can lift assumption~\assumption{1}.

In the worst case our algorithm is as fast as existing algorithms, but in practice our method exploits the geometry of the symmetric difference between input and output shape: it minimizes unnecessary disassembly and reassembly of the particle system. Furthermore, if the configuration of the particle system is such that the feather trees are balanced with respect to the amount of supply particles in each sub-tree, then our algorithm finishes in a number of rounds close to the diameter of the system (instead of the number of particles). Our reconfiguration algorithm also moves the particles over as many parallel shortest paths as the problem instance allows. In Section~\ref{sec:conclusion} we discuss these features in more detail and also sketch future work.

\section{Shortest path trees}\label{sec:SPtrees}
To solve the particle reconfiguration problem, we need to coordinate the movement of the particles from $I\setminus T$ to $T\setminus I$.
Among the previously proposed primitives for amoebot coordination is the \emph{spanning forest  primitive}~\cite{derakhshandeh2015leader} which organizes the particles into trees to facilitate movement while preserving connectivity.
The root of a tree initiates the movement, and the remaining particles follow via handovers between parents and children.
However, the spanning forest primitive does not impose any additional structure on the resulting spanning trees.
We propose to use a special kind of shortest path trees (SP-trees), called \emph{feather trees}, which were briefly introduced in~\cite{disc22}.

To ensure that all paths in the tree are shortest, we need to control the growth of the tree.
One way to do so, is to use a special counting scheme, where each particle computes its distance to the root. However, to fit into the memory requirements of a single particle, they only compute the difference in distance with neighboring particles~\cite{boulinier2008space}.

\begin{lemma}[\cite{boulinier2008space}]\label{lem:boulinier_sp}
Given a connected particle configuration $\PC$ with diameter~$d$, we can create an SP-tree using at most $O(d)$ rounds.
\end{lemma}

\subsection{Efficient SP-trees}\label{subsec:efficientSPtrees}
To create SP-trees more efficiently for simply connected particle systems, i.e. be able to create multiple of them at the same time, we describe a version of the shortest path map (SPM) data structure~\cite{Mitchell1991} on the grid (see Figure~\ref{fig:spm} (left)).
We say that a particle $q\in\PC$ is \emph{$\PC$-visible} from a particle $p\in\PC$ if there exists a shortest path from $p$ to $q$ in $G$ that is contained in $G_\PC$.
This definition of visibility is closely related to staircase visibility 
in rectilinear polygons~\cite{Culberson1987,Ghosh2007}.
Let $\PC$ be simply connected, and let $R_0\subseteq \PC$ be the subconfiguration of all particles $\PC$-visible from some particle~$r$.
By analogy with the geometric SPM, we refer to $R_0$ as a \emph{region}.
If $R_0=\PC$ then $\SPM(r)$ is simply $R_0$.
Otherwise, consider the connected components $\{\RC_1,\RC_2,\dots\}$ of $\PC \setminus R_0$.
The \emph{window} of $\RC_i$ is a maximal straight-line chain of particles in $R_0$, each of which is adjacent to a particle in $\RC_i$ (e.g., in Figure~\ref{fig:spm} (right), chain $(r_i,w)$ is a window).
Denote by~$r_i$ the closest particle to $r$ of the window $W_i$ of $\RC_i$.
Then $\SPM(r)$ is recursively defined as the union of $R_0$ and $\SPM(r_i)$ in $\RC_i \cup W_i$ for all $i$.
Let $R_i \subseteq \RC_i\cup W_i$ be the set of particles $\PC$-visible from $r_i$.
We call $R_i$ the \emph{visibility region} of $r_i$, and $r_i$ the \emph{root} of $R_i$.
Note that by our definition the particles of a window between two adjacent regions of a shortest path map belong to both regions.
\begin{figure*}[t]
	\centering
	\includegraphics{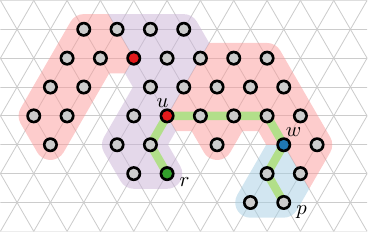}\hfil
	\includegraphics{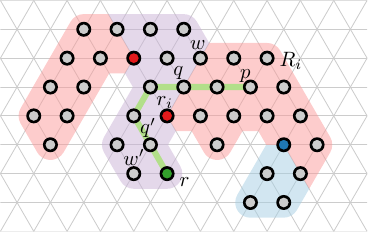}
	\caption{Shortest path map of node $r$. Any shortest path between $r$ and $p$ must pass through the roots of the respective SPM regions ($u$, $w$, and $r_i$). The region $R_0$ (in purple) consists of the particles $\PC$-visible to $r$. The red and the blue particles are the roots of the corresponding SPM regions.
	Right: any path not going through the root of a visibility region can be shortened.}
	\label{fig:spm}
\end{figure*}

\begin{lemma}
Let $r_i$ be the root of a visibility region $R_i$ in a particle configuration $\PC$.
For any particle $p$ in $R_i$, the shortest path from $r$ to $p$ in $\PC$ passes through $r_i$.
\end{lemma}
\begin{proof}
Assume that there exists a shortest path $\pi$ in $G_\PC$ from $r$ to $p$ that does not pass through $r_i$. Assume further that $R_i$ is adjacent to $R_0$ (see Figure~\ref{fig:spm} (right)).
Path $\pi$ must cross window $W_i$ $(r_i,w)$ at some particle $q \neq r_i$.
The extension of window $(r_i,w)$, $(r_i, w')$, partitions $R_0$ into two parts.
Since $R_0$ is a visibility region, $\pi$ must cross $(r_i, w')$ at some particle $q'$.
Now there exists a shorter path from $r$ to $p$: from $r$ to $q'$ to $q$ to $p$.
Contradiction. The same argument applies recursively to regions $R_i$ further removed from $R_0$.
\end{proof}

\begin{corollary}\label{the:spm}
Any shortest path $\pi$ between $r$ and any other particle $p$ in $\PC$ must pass through the roots of the SPM regions that $\pi$ crosses.
\end{corollary}
If a particle $p$ is $\PC$-visible from $r$ then there is a \emph{$60^\circ$-angle monotone path}~\cite{dehkordi2014} from $r$ to $p$ in $G_\PC$.
That is, there exists a $60^\circ$-cone in a fixed orientation, such that for each particle $q$ on the path from $q$ to $p$ lies completely inside this cone translated to $q$ (see Figure~\ref{fig:cones} (left)). 

\begin{figure*}[t]
    \centering
    \includegraphics{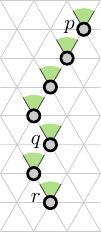}
    \hfil
    \includegraphics[page=1]{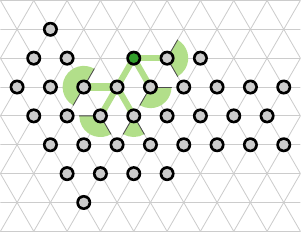}
    \hfil
    \includegraphics[page=2]{angle_trees_non-convex.pdf}
    \caption{Left: An angle monotone path from $r$ to $p$. For every particle~$q$, the remainder of the path lies in a $60^\circ$-cone. Middle: Growing an SP-tree using cones of directions. The particle on the left just extended its cone to $180^\circ$. Right: A couple activations later.}
    \label{fig:cones}
\end{figure*}

We use a version of such cones to grow an SP-tree efficiently.
Each node that is already included in the tree carries a cone of valid growth directions (see Figure~\ref{fig:cones} (middle)).
When a leaf of the tree is activated it includes any neighbors into the tree which are not part of the tree yet and lie within the cone.
A cone is defined as an interval of ports. 
The cone of the root $r$ contains all six ports. 
When a new particle $q$ is included in the tree, then its parent $p$ assigns a particular cone of directions to $q$. Assume parent $p$ has cone $c$ and that $q$ is connected to $p$ via port $i$ of $p$. By definition $i\in c$, since otherwise $p$ would not include $q$ into the tree. 
We intersect $c$ with the $120^\circ$-cone $[i-1, i+1]$ and pass the resulting cone $c'$ on to $q$. (Recall that the arithmetic operations on the ports are performed modulo $6$.)
When doing so we translate $c'$ into the local coordinate system of $q$ such that the cone always includes the same global directions.
This simple rule for cone assignments grows an SP-tree in the visibility region of the root $r$ and it does so in a linear number of rounds.

\begin{lemma}\label{lem:convex}
Given a particle configuration $\PC$ with diameter $d$ which is $\PC$-visible from a particle $r\in\PC$, we can grow an SP-tree in $\PC$ from $r$ using $O(d)$ rounds.
\end{lemma}
\begin{proof}
We are growing the SP-tree from $r$ using the algorithm described above.
Every round, the particles having a cone grow the tree; therefore, every particle is included in the tree after $O(n)$ rounds.
Observe that any path in $G$ which uses only two adjacent directions from the possible six, is a shortest path. Moreover, every shortest path between two particles $p$ and $q$ uses the same at most two directions.
Consider the path $\pi$ in the SP-tree from $r$ to a particle $p$.
By construction, all particles on $\pi$, except for $r$, have cones of at most $120^\circ$ (which equals three adjacent directions).
Assume that $\pi$ leaves $r$ in the global direction $d_1$.
As soon as $\pi$ deviates from $d_1$ in direction $d_2$ adjacent to $d_1$, the cone of available directions shrinks to $60^\circ$ and contains only the two directions $d_1$ and $d_2$.
Hence $\pi$ uses at most two adjacent directions and is therefore a shortest path.

It remains to show that for any particle $q$ in $\PC$ there is a path in the SP-tree from $r$ to $q$.
Suppose that $q$ is not part of the SP-tree.
Consider a shortest path $\sigma$ from $r$ to $q$ in $\PC$, let $p$ be the last particle on this path that belongs to the SP-tree, and $q'$ be the next particle along $\sigma$.
We argue that $q'$ lies in one of the directions of the cone of particle $p$.
Indeed, as observed above, path $\sigma$ consists of at most two adjacent directions.
The path $\pi$ from $r$ to $p$ in the SP-tree uses the same or a subset of the directions of $\sigma$, and thus the cone of $p$ contains the direction towards $q'$.
Contradiction.
\end{proof}

We now extend this solution to arbitrary simply connected particle systems using the shortest-path map $\SPM(r)$. 
The SP-tree constructed by the algorithm above contains exactly the particles of the visibility region $R_0$ of $\SPM(r)$.
As any shortest path from a window particle to $r$ passes through the root of that window, any window incident to $R_0$ forms a single branch of the SP-tree.
To continue the growth of the tree in the remainder of $\PC$, we extend the cone of valid directions for the root particles of the regions of $\SPM(r)$ by $120^\circ$.
A particle $p$ can detect whether it is a root of an $\SPM(r)$-region by checking its local neighborhood.
Specifically, let the parent of $p$ lie in the direction of the port $i+3$.
If (1) the cone assigned to $p$ by its parent is $[i-1,i]$ (or $[i,i+1]$), (2) the neighboring node of $p$ in the direction $i+2$ (or $i-2$) is empty, and (3) the node in the direction $i+1$ (or $i-1$) is not empty, then $p$ is the root of an $\SPM(r)$-region, and thus $p$ extends its cone to $[i-1,i+2]$ (or $[i-2,i+1]$) (see Figure~\ref{fig:cones} (right)).
Note that an extended cone becomes a $180^\circ$-cone.

\begin{lemma}
\label{lem:tree_hits_root}
Let $\SPM(r)$ be the shortest-path map of a particle $r$ in a simply connected particle system $\PC$.
A particle $u\in\PC$ extends its cone during the construction of an SP-tree if and only if it is the root of a region in $\SPM(r)$.
\end{lemma}
\begin{proof}
Let $w$ be a root of an $\SPM(r)$-region. 
We first argue that $w$ extends its cone when it is a leaf of a growing SP-tree and is activated.
Let $w$ lie in an $\SPM(r)$-region with root $u$ (as in Figure~\ref{fig:spm}).
Any shortest path from $u$ to $w$ in $\PC$ uses exactly two adjacent directions, one of which is directed along the window of $w$.
Otherwise either $w$ and the particles of the window would not be $\PC$-visible to $u$, or all neighbors of $w$ would be $\PC$-visible to $u$.
W.l.o.g., let these two directions be $i$ and $i-1$ in the local coordinate system of $w$, and let $i$ be the direction along the window.
All shortest paths from $w$ to $u$ must go in direction $i+3$, otherwise again all neighbors of $w$ would be $\PC$-visible to $u$.
Then, $w$ must have a neighboring particle in the direction of $i+1$, and the neighboring node in the direction $i+2$ must be empty.
Otherwise, $w$ would not be part of a window, or it would not be a root of the window.
Thus, $w$ extends its cone.

Let $u$ be a leaf of a growing SP-tree which extends its cone. We now argue that $u$ must be a root of an $\SPM(r)$-region.
Let $R_0$ be the visibility region of $r$ in $\PC$. We first assume that $u\in R_0$.
Let the parent $p$ of $u$ lie in the direction of port $i+3$.
W.l.o.g., assume that $u$ extends its cone $[i-1,i]$ to $[i-1,i+2]$, and thus the neighboring node at port $i+1$ is non-empty, and the neighboring node at port $i+2$ is empty.
Let $v$ be the neighbor of $u$ in the direction $i+1$. 
Consider a maximal chain of particles $W'$ in $\PC$ from $u$ in the direction of port $i$.
Particles $W'$ are $\PC$-visible from $r$ as a shortest path from $r$ to any $q\in W'$ uses only two directions $i-1$ and $i$.
Particle $v$ is not $\PC$-visible from $r$, as any path from $r$ to $v$ must cross $W'$, and thus use an extra direction $i+1$ or $i+2$.
Consider the connected component $\RC_v$ of $\PC \setminus R_0$ containing $v$.
Since $u$ is adjacent to $v$ it is part of some window $W$ of $\RC_v$.
The parent $p$ of $u$ is not adjacent to $\RC_v$. 
Since all shortest paths from $r$ to a particle in $W$ pass through $u$, $u$ must be the root of $W$. 
Since $u$ lies on the boundary of $\RC_v$, growing the SP-tree further from $u$ is equivalent to growing a new SP-tree only in $\RC_v$ with $u$ as the root. Hence the same argument applies recursively.
\end{proof}

Lemmas~\ref{lem:convex} and~\ref{lem:tree_hits_root} together imply Theorem~\ref{the:SP-tree}.

\begin{theorem}\label{the:SP-tree}
Given a simply connected particle configuration $\PC$ with diameter $d$ and a particle $r\in\PC$ we can grow an SP-tree in $\PC$ from $r$ using $O(d)$ rounds.
\end{theorem}

 Note that this version of SP-tree is created using the same amount of rounds as the one in \cite{boulinier2008space} (Lemma~\ref{lem:boulinier_sp}), and only works for simply connected configurations.
 Nevertheless, we can use this version to create trees that we can overlap.

\subsection{Feather trees}\label{subsec:feathertrees}
These SP-trees, are not unique: the exact shape of the tree depends on the activation sequence of its particles. 
Our approach to the reconfiguration problem is to construct multiple overlapping trees which the particles use to navigate across the structure. 
As the memory capacity of the particles is restricted, they cannot distinguish between multiple SP-trees by using ids.
Thus we need SP-trees that are unique and have a more restricted shape, so that the particles can distinguish between them by using their geometric properties.
In this section we hence introduce \emph{feather trees} which are a special case of SP-trees that use narrower cones during the growth process.
As a result, feather trees bifurcate less and have straighter branches.

Feather trees follow the same construction rules as our previous SP-trees, but with a slightly different specification of cones.
We distinguish between particles on \emph{shafts} (emanating from the root or other specific nodes) and \emph{branches} (see Figure~\ref{fig:feather_tree} (left)).
The root $r$ chooses a maximal independent set of neighbors $N_\mathit{ind}$; it contains at most three particles and there are at most two ways to choose.
The particles in $N_\mathit{ind}$ receive a standard cone with three directions (a 3-cone), and form the bases of shafts emanating from $r$.
All other neighbors of $r$ receive a cone with a single direction (a 1-cone), and form the bases of branches emanating from $r$.
For a neighbor~$p$ across the port $i$, $p$ receives the cone $[i-1,i+1]$, translated to the coordinate system of $p$, if $p$ is in $N_\mathit{ind}$, and the cone $[i]$ otherwise.
The shaft particles propagate the 3-cone straight, and 1-cones into the other two directions, thus starting new branches.
Hence all particles (except for, possibly, the root) have either a 3-cone or a 1-cone.
The particles with 3-cones lie on shafts and the particles with 1-cones lie on branches.

We extend the construction of the tree around reflex vertices on the boundary of $\PC$ in a similar manner as before.
If a branch particle $p$ receives a 1-cone from some direction $i+3$, and the direction $i+2$ (or $i-2$) does not contain a particle while the direction $i+1$ (or $i-1$) does, then $p$ 
initiates a growth of a new shaft in the direction $i+1$ (or $i-1$) by sending there a corresponding 3-cone (see Figure~\ref{fig:feather_tree} (right)).

\begin{figure*}[t]
\centering
    \includegraphics{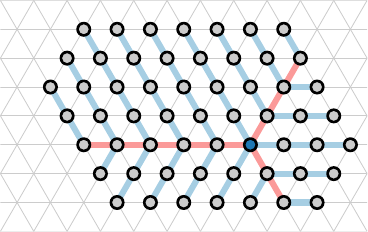}
    \hfil
    \includegraphics{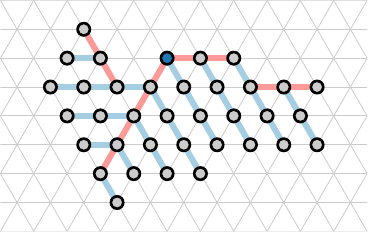}
    \caption{
   Two feather trees growing from the dark blue root. Shafts are red and branches are blue. Left: every particle is reachable by the initial feathers; Right: additional feathers are necessary.
}
    \label{fig:feather_tree}
\end{figure*}

Feather trees are a more restricted version of SP-trees (Theorem~\ref{the:SP-tree}).
For every feather tree, there exists an activation order of the particles such that the SP-tree algorithm would create this specific feather tree.
This leads us to the following lemma.

\begin{lemma}\label{lem:feather_trees}
Given a simply connected particle configuration $\PC$ with diameter $d$ and a particle $r\in \PC$, we can grow a feather tree from $r$ in $O(d)$ rounds.
\end{lemma}
Every particle is reached by a feather tree exactly once, from one particular direction.
Hence a feather tree is independent of the activation sequence of the particles.
In the following we describe how to navigate a set of overlapping feather trees.
To do so, we first identify a useful property of  shortest paths in feather trees.

We say that a vertex $v$ of $G_\PC$ is an \emph{inner vertex}, if $v$ and its six neighbors lie in the core $I\cap T$. 
All other vertices of the core are \emph{boundary vertices}. 
A \emph{bend} in a path is formed by three consecutive vertices that form a $120^\circ$ angle. 
We say that a bend is an \emph{boundary bend} if all three of its vertices are boundary vertices; otherwise the bend is an \emph{inner bend}.

\begin{definition}[Feather Path]
A path in $G_\PC$ is a \emph{feather path} if it does not contain two consecutive inner bends.
\end{definition}
We argue that every path $\pi$ from the root to a leaf in a feather tree is a feather path.
This follows from the fact that inner bends can occur only on shafts, and $\pi$ must alternate visiting shafts and branches.

\begin{lemma}\label{lem:feather_trees_bends}
    A path between a particle $s$ and a particle $t$ is a feather path if and only if it lies on a feather tree rooted at $s$.
\end{lemma}
\begin{proof}
    To show that every path on a feather tree is a feather path, consider path $\pi$ from root $s$ to an arbitrary particle $t$ on the feather tree rooted at $s$.
    Assume for contradiction that there are two consecutive inner bends on $\pi$.
    Only bends on a shaft (shaft bends) can be inner bends.
    Hence $\pi$ must contain two consecutive shaft bends. However, at a shaft bend the path moves on to a branch. The path cannot leave this branch without a bend, which is necessarily a boundary bend. Contradiction.

    To show that every feather path lies on a feather tree, consider a feather path $\pi$ from some particle $s$ to an arbitrary particle $t$.
    Let the direction that $\pi$ leaves $s$ be the shaft of the tree rooted at $s$.
    If $\pi$ makes a bend, the feather tree makes the same bend, going from a shaft to a branch.
    Now $\pi$ cannot make another bend unless it is a boundary bend.
    If it does make a boundary bend, the feather tree makes this boundary bend by going from a branch to a shaft.
    The feather tree follows $\pi$ for every bend it makes and therefore $\pi$ is part of a feather tree rooted at $s$.    
\end{proof}

\subparagraph{Navigating feather trees}
Consider a directed graph composed of multiple overlapping feather trees with edges pointing from roots to leaves.
Due to its limited memory, a particle cannot store the identity of the tree it is currently traversing.
Despite that, particles can navigate down the graph towards the leaf of some feather tree and remain on the correct feather tree simply by counting the number of inner bends and making sure that the particle stays on a feather path (Lemma~\ref{lem:feather_trees_bends}).
The number of inner bends can either be zero or one, and thus storing this information does not violate the assumption of constant memory per particle.
Thus, when starting at the root of a feather tree, a particle $p$ always reaches a leaf of that same tree.
In particular, it is always a valid choice for $p$ to continue straight ahead (if feasible).
A left or right $120^\circ$ turn is a valid choice if it is a boundary bend, or if the last bend $p$ made was boundary.

When moving against the direction of the edges, up the graph towards the root of some tree, we cannot control which root of which feather tree a particle $p$ reaches, but it still does so along a shortest path.
In particular, if $p$'s last turn was on an inner bend, then its only valid choice is to continue straight ahead.
Otherwise, all three options (straight ahead or a $120^\circ$ left or right turn) are valid.

\section{Supply and demand}\label{sec:supply-demand}
Each supply root organizes its supply component into an SP-tree; the supply particles will navigate through the supply roots into the core $I\cap T$ and towards the demand components along a \emph{supply graph}.
The supply graph, constructed in $I\cap T$, serves as a navigation network for the particles moving from the supply to the demand along shortest paths.
Let $G_{I\cap T}$ be the subgraph of $G$ induced on the nodes of $I\cap T$.
We say a supply graph $S$ is a subgraph of $G_{I\cap T}$ connecting every supply root $s$ to every demand root $d$ such that the following three \emph{supply graph properties} hold:
\begin{enumerate}
    \item for every pair $(d, s)$ a shortest path from $d$ to $s$ in $S$ is also a shortest path in $G_{I\cap T}$,
    \item for every pair $(d, s)$ there exists a shortest path from $d$ to $s$ in $S$ that is a feather path,
    \item every particle~$p$ in $S$ lies on a shortest path for some pair $(d, s)$.
\end{enumerate}

We orient the edges of $S$ from demand to supply, possibly creating parallel edges oriented in opposite directions.
For a directed edge from $u$ to $v$ in $S$, we say that $u$ is the predecessor of $v$, and $v$ is the successor of $u$.

To create the supply graph satisfying the above properties, we use feather trees rooted at demand roots.
Every demand root initiates the growth of its feather tree.
When a feather tree reaches a supply root, a \emph{supply found token} is sent back to the root of the tree.
Note that if several feather trees overlap, a particle $p$ in charge of forwarding the token up the tree cannot always determine which specific direction the corresponding root of the tree is.
It thus sends a copy of the token to all \emph{valid} parents (predecessors of $p$ on every possible feather path of the token), and so the token eventually reaches all demand roots reachable by a feather path from the node it was created in.
To detect if a token has already made an inner bend, the supply found token carries a flag $\beta$ that is set once the token makes an inner bend.
Specifically, a particle~$p$ that receives a supply found token $t$ does the following:

\begin{enumerate}
\item $p$ marks itself as part of the supply graph $S$, and adds the direction $i$ that $t$ came from as a valid successor in $S$, 
\item from the set of all its predecessors (for all incoming edges), $p$ computes the set $U$ of valid predecessors. When $p$ is not a reflex vertex, $U=[i+3]$ if $\beta=1$, and $U=[i+2,i+4]$ if $\beta=0$. For reflex vertices, $\beta$ is reset to $0$ and more directions might be valid. Particle~$p$ then adds $U$ to the set of its valid predecessors in $S$,
\item $p$ sends copies of $t$ (with an updated value of $\beta$) to the particles in $U$, 
and
\item $p$ stores $t$ in its own memory.
\end{enumerate}

\noindent Each particle~$p$, from each direction $i$, stores at most one token with flag $\beta$ set to $1$ and one with $\beta$ set to $0$.
Hence $p$ can store the corresponding information in its memory.
If a particle $p$ already belongs to the supply graph $S$ when it is reached by a feather tree $F$, then $p$ checks if its predecessor $q$ in $F$ is a valid parent for any token~$t$ stored in $p$, and, if so, adds $q$ to the set of its valid predecessors in $S$ (as in step~(2)).
For each of these tokens, but for at most one for each values of $\beta$, $p$ sends a copy of $t$ to $q$ (as in step (3)).

\begin{lemma}\label{lem:supply_graph}
Given a simply connected particle configuration $\PC$ with diameter $d$, a set of particles marked as supply roots, and a set of particles marked as demand roots, a supply graph can be constructed in $O(d)$ rounds.
\end{lemma}
\begin{proof}
The graph $S$ constructed via the steps listed above satisfies the three supply graph properties by construction.
The algorithm which constructs the supply graph grows a feather tree from every demand root $d$ in parallel.
Hence, by Lemma~\ref{lem:feather_trees} every particle becomes a part of every tree in $O(n)$ rounds.
When a supply root $s$ is reached by a tree $F_d$, a supply found token $t$ is sent back to $d$ and it reaches $d$ in $O(n)$ rounds.
Every supply root~$s$ is found by every feather tree $F$ either directly, or indirectly by reaching a particle $p$ that is already in $S$ and that contains a token $t$ coming from $s$ for which $F$ is a valid predecessor.
Thus, the supply graph is constructed in $O(n)$ rounds in total.
\end{proof}


\subparagraph{Bubbles}
\begin{figure}[b]
\centering
    \includegraphics{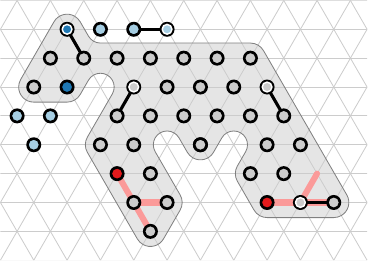}
    \caption{
   A reconfiguration process with five expanded particles holding bubbles (outlined in white). Supply particles are blue, supply roots dark blue. Demand roots (red) store spanning trees of their demand components. The supply graph is not shown.}
    \label{fig:bubbles}
\end{figure}

Particles move from supply to demand.
However, for ease of presentation and analysis, we introduce the abstract concept of demand \emph{bubbles} that move from demand to supply, in the direction of edges of $S$, see Figure~\ref{fig:bubbles}.
Let us assume for now that the supply graph $S$ has been constructed (in fact, its construction can occur in parallel with the reconfiguration process described below).
Starting with a corresponding demand root $d$, each demand component $D$ is constructed by particles flowing from the core $I \cap T$, according to the spanning tree of $D$ stored in $d$.
Every time a leaf particle expands in $D$, it creates a \emph{bubble} of demand that needs to travel through $d$ down $S$ to the supply, where it can be resolved.
Bubbles move via a series of handovers along shortest paths in $S$.
An expanded particle~$p$ holding a bubble~$b$ stores two values associated with it.
The first value $\beta$ is the number of inner bends $b$ took since the last boundary bend ($\beta\in\{0,1\}$), and is used to route the bubbles in $S$.
The second value $\delta$ stores the general direction of $b$'s movement; $\delta=\textsc{s}$ if $b$ is moving forward to supply, and $\delta=\textsc{d}$ if $b$ is moving backwards to demand.

If a particle $p$ holding a bubble~$b$ wants to move $b$ to a neighboring particle $q$, $p$ can only do so if $q$ is contracted.
Then, $p$ initiates a pull operation, and thereby transfers $b$ and its corresponding values to $q$.
Thus the particles are pulled in the direction of a demand root, but the bubbles travel along $S$ from a demand root towards the supply.

A supply component may become empty before all bubbles moving towards it are resolved.
In this case, the particles of $S$ have to move the bubbles back up the graph.
Particles do not have sufficient memory to store which specific demand root bubbles came from.
However, because of the first supply graph property, every demand root has a connection to the remaining supply.
While a bubble is moved back up along $S$, as soon as there is a different path towards some other supply root, it is moved into that path.
Then, the edges connecting to the now empty supply are deleted from $S$.
Moreover, other edges that now point to empty supply or to deleted edges are themselves deleted from $S$.
As the initial and target shapes have the same size, the total number of bubbles equals the number of supply particles.
Therefore, once all bubbles are resolved, the reconfiguration problem is solved.

In the remainder of the paper we may say ``a bubble activates'' or ``a bubble moves''.
By this we imply that ``a particle holding a bubble activates'' or ``a particle holding a bubble moves the bubble to a neighboring particle by activating a pull handover''.

\section{Navigating the supply graph}\label{sec:nav}
When the demand and supply roots are connected with the supply graph $S$, as described in Section~\ref{sec:supply-demand}, the reconfiguration process begins.
Once a demand root $d$ has received a supply found token from at least one successor ($d$ is added to $S$), it begins to construct its demand component $D$ along the spanning tree $\TC_D$ that $d$ stores, and starts sending demand bubbles into $S$.
The leafs of the partially built spanning tree $\TC_D$ carry the information about their respective sub-trees yet to be built, and pull the chains of particles from $d$ to fill in those sub-trees, thus generating bubbles that travel through $d$ into $S$.

With each node $v$ of $S$, for every combination of the direction $i$ to a predecessor of $v$ and a value of $\beta$, we associate a value $\lambda(i,\beta)\in\{\text{true},\text{false}\}$, which encodes the liveliness of feather paths with the corresponding value of $\beta$ from the direction $i$ through $v$ to some supply.
If $\lambda(i,\beta)=\text{false}$ then, for a given value of $\beta$, there are no feather paths through $v$ to non-empty supply in $S$.
Note that here we specifically consider nodes of $S$, and not the particles occupying them.
When particles travel through $S$, they maintain the values of $\lambda$ associated with the corresponding nodes.
Initially, when $S$ is being constructed, $\lambda=\text{true}$ for all nodes, all directions, and all $\beta$.
When a bubble $b$ travels down $S$ to a supply component that turns out to be empty, $b$ reverses its direction.
Then, for all the nodes that $b$ visits while reversing, the corresponding value $\lambda$ is set to $\text{false}$, thus marking the path as dead.

For an expanded particle $p$ occupying two adjacent nodes of $S$, denote the predecessor node as $v_a$, and the successor node as $v_b$.
By our convention, we say the bubble in $p$ occupies $v_b$.
When particle~$p$ with a bubble~$b$ activates, it performs one of the following operations.
It checks them in order and performs the first action available.
\begin{enumerate}
    \item If $\delta=\textsc{s}$ ($b$ is moving to supply) and $v_b$ is inside a supply component, then $p$ pulls on any contracted child in the spanning tree of the supply; if $v_b$ is a leaf, then $p$ simply contracts into $v_a$, thus resolving the bubble.
    \item If $\delta=\textsc{s}$ and $v_b\in S$, $p$ checks which of the successors of $v_b$ in $S$ lie on a feather path for $b$ and are alive (i.e., their corresponding  $\lambda=\text{true}$).
    If there is such a successor $q$ that is contracted, $p$ pulls $q$ and sends it the corresponding values of $\beta$ and $\delta$ (while updating $\beta$ if needed), i.e., $p$ transfers $b$ to $q$.
    Thus the bubble moves down $S$ to supply.
    \item If $\delta=\textsc{s}$, $v_b\in S$, and $v_b$ does not have alive successors in $S$ that lie on a feather path for~$b$, then $p$ reverses the direction of $b$ to $\delta=\textsc{d}$, and sets the value $\lambda(i,\beta)$ of $v_b$ to $\text{false}$, where $i$ is the direction from $v_b$ to $v_a$. 
    The bubble does not move.
    \item If $\delta=\textsc{d}$ ($b$ is moving to demand) and there exists an alive successor node of either $v_b$ or $v_a$ that lies on a feather path for $b$ and is occupied by a contracted particle $q$, then $p$ switches the direction of $b$ to $\delta=\textsc{s}$, pulls on $q$, and transfers to $q$ the bubble $b$ with its corresponding values (while updating $\beta$ if needed). The bubble changes direction and moves onto a feather path which is alive.
    \item If $\delta=\textsc{d}$ and none of the successors of $v_b$ or $v_a$ in $S$ are alive for $b$, then $p$ sets the corresponding values $\lambda$ of $v_b$ and $v_a$ to $\text{false}$, and checks which predecessors of $v_a$ lie on a feather path for $b$ (note, this set is non-empty).
    If there is such a predecessor $q$ that is contracted, then $p$ pulls $q$ and transfers it the bubble $b$.
    Thus the bubble moves up in~$S$.
\end{enumerate}

\subsection{Coarse grid}
Particles can only make progress if they have a contracted successor.
If there are crossing paths in $S$, these paths might interfere.
To ensure that flows of particles along different feather paths can cross, we introduce a coarsened grid, and devise a special crossing procedure.

Rather than constructing supply graph $S$ on the triangular grid $G$, we now do so using a grid~$G_L$ that is coarsened by a factor of three and is overlaid over the core $I\cap T$ (see Figure~\ref{fig:particle_reconfiguration_example}).
Then, among the nodes of $G$ we distinguish between those that are also \emph{grid nodes} of $G_L$, \emph{edge nodes} of $G_L$, and those that are neither.
By our assumptions on the input, the graph $G_L$ is connected.
Note, that then, every node of the particle system is either a part of $G_L$ (is a grid or an edge node), or is adjacent to a node of $G_L$.
To ensure that all particles agree on the location of $G_L$, we assume that we are given a leader particle $\ell$ in the core $I\cap T$, that initiates the construction of $G_L$ (see Figure~\ref{fig:particle_reconfiguration_example} in Section~\ref{sec:problem_description}).

Note that for two adjacent grid nodes $v_1$ and $v_2$ of $G_L$, the edge $(v_1, v_2)$ does not have to be in $G_L$, if one or both corresponding edge nodes are not part of the core $I \cap T$.
We say a grid node~$v$ in $G_L$ is a boundary node if there exists a grid cell with corners $v$, $v_1$, $v_2$ where at least one of its edges is missing.
Grid node~$v$ is a \emph{direct boundary} node if this missing edge is incident to $v$, and an \emph{indirect boundary} node if the missing edge is $(v_1, v_2)$.
All other grid nodes in $G_L$ are inner nodes.

Now, feather trees only grow over the particles in $G_L$.
Demand roots initiate the construction of their feather trees from a closest grid node in $G_L$.
The supply roots organize their respective supply components in SP-trees as before (in the original grid), and connect to all adjacent particles in the supply graph $S$ ($S\subseteq G_L$).

The growth of a feather tree in $G_L$ is very similar to that in $G$.
On direct boundary nodes, the cones propagate according to the same rules as before.
On indirect boundary nodes, the creation of new shafts and branches is outlined in Figure~\ref{fig:trees_on_large_grid}.
The resulting feather tree may now have angles of $60^\circ$.
To still be able to navigate the supply graph, bubbles are allowed to make a $60^\circ$ bend on indirect boundary nodes, resetting $\beta=0$, only when both the start and end vertex of that bend are also boundary nodes.

\begin{figure*}[]
\centering
\includegraphics{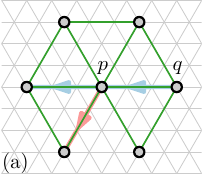}\hfill
\includegraphics{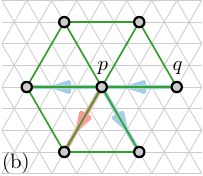}\hfill
\includegraphics{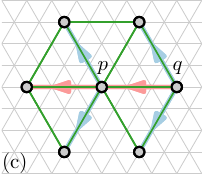}\hfill
\includegraphics{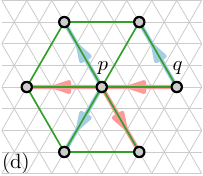}
\caption{Construction of feather trees at indirect boundary nodes.
$G_L$ is shown in green, only the particles on nodes of $G_L$ are shown.
A branch ((a)--(b), in blue) and a shaft ((c)--(d), in red) of a feather tree grows from $q$ to $p$.
(a) A new shaft is emanated from $p$;
(b) a new shaft and a new branch are emanated from $p$;
(c) $p$ behaves as an internal node;
(d) a new shaft is emanated from $p$.}
\label{fig:trees_on_large_grid}
\end{figure*}

A useful property of $G_L$ that helps us ensure smooth crossings of the bubble flows is that there are at least two edge nodes in between any two grid nodes.
We can now restrict the expanded particles carrying bubbles to mainly use the edge nodes of $S$, and only cross the grid nodes if there is enough room for them to fully cross.

With $p(x)$, we denote the particle occupying node $x$.
For a directed edge from $u$ to $v$ in $S$, we say that $u$ is the predecessor of $v$, and $v$ is the successor of $u$, and denote this by $u\rightarrow v$.
Consider a bubble $b$ moving forward in $S$ (with $\delta=\textsc{s}$), held by an expanded particle $p$ occupying two edge nodes $v_a$ and $v_b$, with $v_a\rightarrow v_b$ (see Figure~\ref{fig:large_grid_directions}).
\begin{figure}
    \centering
    \includegraphics{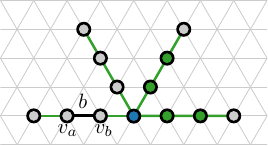}
    \caption{A bubble~$b$ on two edge nodes of $G_L$. Only particles on $G_L$ are drawn. The blue particle is on $c(b)$, the nodes occupied by the two pairs of green particles together form $N(b)$.}
    \label{fig:large_grid_directions}
\end{figure}
Let $c(b)$ be the grid node of $S$ adjacent to $v_b$, such that $v_b \rightarrow c(b)$; $c(b)$ is the grid node that $b$ needs to cross next.
Let $\mathcal{D}(b)$ denote the set of valid directions for $b$ in $S$ from the position of $c(b)$.
Let $N(b)=\{(r^i_1,r^i_2) \mid i\in \mathcal{D}(b)\}$ be the set of pairs of edge nodes lying in these valid directions for $b$.
Specifically, for each $i\in\mathcal{D}(b)$, let $c(b)\rightarrow r^i_1$ in the direction $i$, and $r^i_1\rightarrow r^i_2$.

\subparagraph{Crossings in $G_L$}
The grid nodes of $G_L$ that are part of the supply graph act as traffic conductors.
We thus use terms grid nodes and junctions interchangeably.
For bubble $b$, its particle $p$ is only allowed to pull on the particle at $c(b)$, and thus initiate the crossing of $c(b)$, if there is a pair of nodes $(r_1,r_2)\in N(b)$ that are occupied by contracted particles.
In this case, after at most three activation rounds, $b$ will completely cross $c(b)$, the expanded particle now carrying it will occupy the edge nodes $r_1$ and $r_2$, and the junction $c(b)$ will be ready to send another bubble through itself.
Assume for now that $S$ has all edges oriented in one direction (there are no parallel edges in opposite directions).
Below we discuss how to lift this assumption.
The procedure followed by the junctions is the following.
If an expanded particle $p$ wants to pull on a particle at a junction, it first requests permission to do so by sending a \emph{request} token containing the direction it wants to go after $c(b)$.
Every junction node stores a queue of these requests.
A request token arriving from the port $i$ is only added to the queue if there are no requests from $i$ in the queue yet.
As every direction is stored only once, this queue is at most of size six.
When particle $p(r_1)$ occupying an edge node $r_1$, with some grid node $c\rightarrow r_1$, activates, it checks if itself and the particle $p(r_2)$ at $r_1 \rightarrow r_2$ are contracted.
If so, $p(r_1)$ sends an \emph{availability} token to $p(c)$.
Similarly, junction nodes need to store at most six availability tokens at once.

When the particle at a junction activates, and it is ready to transfer the next bubble, it grants the first pull request with a matching availability token by sending the acknowledgment token to the particle holding the corresponding bubble.
The request and availability tokens are then consumed.
Only particles with granted pull requests are allowed to pull and move their bubbles onto junctions.
Junction queues are associated with the grid nodes of $S$, and not particles.
Thus, if an expanded particle $p$ occupying a grid node~$c$ pulls another particle~$q$ to $c$, the queues are transformed into the coordinate system of $q$ and sent to $q$.

\begin{remark}
Above, for simplicity of presentation, we assume that there are no bidirectional edges in $S$.
These edges, however, can be treated as follows.
During the construction of the supply graph one of the two opposite directions is chosen as a dominating one.
For a corresponding edge of $G_L$, one of its two edge nodes that receives the supply found token first, reserves the direction of the corresponding feather path for this edge.
If the other edge node eventually receives the supply found token from the opposite direction, it stores the information about this edge as being inactive.
While the branch of the dominating direction is alive, the dominated branch is marked as unavailable.
As soon as, and if, supply runs out for the dominating branch, and it is marked dead, the dominated branch is activated, and can now be used.
This slows down the reconfiguration process by a number of rounds at most linear in the size of the dominating branch.
\end{remark}

\section{Algorithm}\label{sec:algo}
To summarize, our approach consists of three phases.
In the first phase, the leader particle initiates the construction of the coarse grid $G_L$ over the core $I\cap T$.
In the second phase, the particles grow feather trees, starting from the demand roots.
If a feather tree reaches supply, that information is sent back up the tree and the particles form the supply graph~$S$.
In the last phase, particles move from supply to demand along $S$.
Note that, for the particle system as a whole, these phases may overlap in time.
For example, the reconfiguration process may begin before the supply graph is fully constructed.
Each individual particle can move on to executing the next phase of the algorithm once the previous phase for it is finished.

For the purposes of analysis, we view the reconfiguration as bubbles of demand traveling along $S$ from demand to supply.
Bubbles turn around on dead paths where all supply has been consumed.
To ensure proper crossing of different bubble flows, we let the grid nodes of $S$ to act as traffic conductors, letting some bubbles cross while others wait for their turn.

\subsection{Correctness.}
For simplicity of presentation, we first show correctness of the algorithm under a sequential scheduler.
We then extend the algorithm and its analysis to the case of an asynchronous scheduler.
To show correctness, we need to show two properties, \emph{safety} and \emph{liveness}.
The algorithm is safe if $\PC$ never enters an invalid state, and is live if in any valid state there exists a particle that, when activated, can make progress towards the goal.
Lemma~\ref{lem:supply_graph} proves the correctness of the phase of the construction of the supply graph.
As the particles start executing the reconfiguration phase only after the construction of the supply graph is finished for them locally, for the purposes of proving the correctness of the reconfiguration phase, we may assume that the supply graph has been constructed in the particle system as a whole.
A state of $\PC$ is valid when it satisfies the following properties:
\begin{itemize}
    \item Particle configuration $\PC$ is connected.
    \item It holds that $\#b + \#d = \#s$, where~$\#b$, $\#d$, and $\#s$ are the number of bubbles, demand spots, and supply particles respectively. That is, the size of supply matches exactly the size of demand.
    This assumes that initial and target shapes have the same size.
    \item There are no bidirectional edges in $S$ allowed for traversal.
    \item For every pair of demand root $d$ and supply root $s$ with a non-empty supply, there exists a feather path from $d$ to $s$ in $S$; furthermore, all such feather paths are alive (i.e., the corresponding values of $\lambda$ are set to true).
    \item Any expanded particle has both nodes on a single feather path of $S$.
    \item Any expanded particle on an alive feather path moves to supply ($\delta=\textsc{s}$).
    \item Any node of $S$ with $\lambda(i,\beta)=\text{true}$, for some combination of direction $i$ and value $\beta$, is connected by an alive feather path to some demand root $d$.
\end{itemize}

For the property of liveness, we need to show that progress can always be made.
We say the particle system makes progress whenever (1) a bubble moves on its path (this includes resolving the bubble with supply, or when a demand root creates a new bubble), and (2) a bubble changes its value for direction $\delta$.
We show by induction that during the execution of our algorithm, the configuration stays valid at any point in time, and that at any point in time, there is a bubble that can make progress.

\begin{lemma}\label{lem:connected}
	At any time, $\PC$ forms a single connected component.
\end{lemma}
\begin{proof}
Initial shape $I$ is connected.
A connected particle configuration can only become disconnected by the sole contraction of a particle.
Assume for contradiction that $\PC$ becomes disconnected by the sole contraction of a particle~$p$, then $p$ must be supply.
Let particle~$q$ be a neighbor of $p$ that becomes disconnected from $p$ and let $r$ be the supply root of the connected component of supply of $p$.
If $q$ is supply, it is in the same supply component as $p$ and therefore there exists a path from $q$ to $r$ to $p$.
If $q$ is not supply, it is in the core.
The core is always connected.
Therefore there also exists a path from $q$ to $r$ to $p$.
This is a contradiction, which proves the lemma.
\end{proof}

\begin{lemma}\label{lem:bubble_number}
For the number of bubbles~$\#b$, the number of supply particles left~$\#s$ and the number of demand spots left~$\#d$, it always holds that $\#b + \#d = \#s$.
\end{lemma}
\begin{proof}
We prove the lemma by induction.
Initially, $\#b = 0$ and $\#s = \#d$, so also $\#b + \#d = \#s$.
Assume that $\#b + \#d = \#s$ for some configuration.
Now consider the activation of a single particle.
We only consider activations that change any of $\#b$, $\#d$, or $\#s$.
They change only in one of two scenarios.
Either a particle gets expanded into a previously empty node of a demand component, creating a bubble, or a bubble gets resolved when a supply particle is contracted.
In the first case we have $\#b + 1 + \#d - 1 = \#s$.
In the second case we have $\#b - 1 + \#d = \#s - 1$.
\end{proof}

\begin{observation}\label{obs:bidirectional}
By construction, for two parallel edges in $S$, both of them cannot be active at the same moment in time.
\end{observation}

\begin{lemma}\label{lem:single_branch}
At any time, every expanded particle has both nodes on a single branch of supply graph $S$.
\end{lemma}
\begin{proof}
We prove the lemma by induction.
The initial shape $I$ has no expanded particles.
Assume that all bubbles are on a single branch of supply graph $S$.
A new bubble can only be created by a demand root and is by definition located on a single branch of $S$.
If a bubble moves, it always moves in such a way that it either stays on the branch it is on, or it changes completely from one branch into another.
Therefore after any activation the lemma still holds.
\end{proof}

\begin{lemma}
At any time, any expanded particle on an alive feather path moves to supply ($\delta=\textsc{s}$).
\end{lemma}
\begin{proof}
We prove the lemma by induction.
The initial shape $I$ has no expanded particles.
Assume that at some moment in time all bubbles in $I\cap T$ that are on alive feather paths move to supply.
A bubble turns its direction back to demand only when there are no alive feather paths to supply from its current node.
Therefore after any activation the lemma still holds.
\end{proof}

\begin{lemma}\label{lem:alive-path}
At any time, any node of $S$ with $\lambda(i,\beta)=$true, for some combination of direction $i$ and value $\beta$, is connected by an alive feather path to some demand root $d$.
\end{lemma}
\begin{proof}
We prove the lemma by induction.
Initially, for any node on $S$, it lies on a live feather path to some demand root.
Assume that at some moment in time every node $v$ of $S$ that is alive for some direction $i$ to its predecessor and some value of $\beta$ is connected to some demand node $d$.
Consider the corresponding feather path $\pi$ connecting $d$ to $v$.
Let $p$ be the particle currently occupying node $v$.
By the previous lemma, $p$'s bubble is moving to supply.
If, in the following activation, $p$ reverses to the demand, it marks $v$ as dead for these specific values of $i$ and $\beta$, and thus $v$ is not alive anymore.
The nodes of $\pi$ cannot be marked as dead, while $v$ is still marked as alive.
Assume that this can occur.
Consider the closest node $u$ to $v$, which is an ancestor of $v$ in $\pi$, which is marked by some particle $q$ as dead.
Then $u$ has a successor along $\pi$ which is alive for $i$ and $\beta$.
But then, $q$ would check if $u$ has any alive successors before marking it as dead.
And thus, all the nodes of $\pi$ connecting $v$ to $d$ must remain alive.
Therefore after any activation the lemma still holds.
\end{proof}

\begin{lemma}\label{lem:valid_supply_graph}
	At any time, for a pair of demand root $d$ and a supply root $s$ of a non-empty supply, all feather paths connecting $d$ and $s$ in $S$ are alive.
\end{lemma}
\begin{proof}
	The existence of a feather path in $S$ connecting $d$ and $s$ follows from the assumption for the purposes of the proof of correctness that $S$ has been fully constructed.
	We prove that all such feather paths are alive by induction.
	Initially, all nodes on $S$ are alive.
	Assume that at some time all feather paths connecting $d$ to $s$ in $S$ are alive.
	Similarly to the proof of Lemma~\ref{lem:alive-path}, a node on a feather path $\pi$ from $d$ to $s$ can be marked dead by a particle only if that node does not have a live successor in $\pi$.
	Thus an interior node of $\pi$ cannot be marked dead.
	The end of $\pi$, i.e., the supply root $s$, is only marked dead once the supply is empty.
	Therefore, as long as the supply of $s$ is non-empty, all the feather paths from $d$ to $s$ are alive.
\end{proof}

Next, using the proven above safety property, we show liveness by arguing that progress can always be made.
We say the particle system makes progress whenever (1) a bubble moves on its path (this includes resolving the bubble with supply, or when a demand root creates a new bubble), and (2) a bubble changes its value for direction~$\delta$.

\begin{lemma}\label{lem:bubble_progress}
If there exist bubbles in the system, some bubble eventually makes progress.
\end{lemma}
\begin{proof}
If there exists a bubble in a supply component, consider the closest such bubble $b$ to a leaf of the supply spanning tree.
If $b$ is at a leaf, the particle carrying it contracts in the next round, thus resolving the bubble.
If $b$ is not a leaf, its node has a contracted child, which will be pulled by the particle of $b$ in the next round, thus making a progress.

If there are no bubbles in the supply components, but there exists a bubble in the supply graph, consider the closest bubble $b$ in $S$ to a supply root of a non-empty supply component.
By the safety property, there is a live feather path from $b$ to the supply root.
If the particle $p$ carrying $b$ occupies one of the grid nodes of $G_L$, then there must be a contracted successor in a valid direction for $p$ to pull on in the next round, otherwise the crossing would not have been initiated.
If $p$ occupies edge nodes of $G_L$, it will send a request token to the grid node $c(b)$ to initiate the crossing.
In the following round, the grid node $c(b)$ will either grant the permission to $b$ to cross it, or it will grant access to another bubble $b'$ such that $c(b')=c(b)$.
This implies that in the following round either $b$ or $b'$ makes a progress.

Finally, if there are no bubbles in the supply components nor the supply graph, but there are bubbles in the demand components, by a similar argument we can show that one of the closest bubbles to a demand root will make a progress in the following round.
\end{proof}

\begin{lemma}\label{lem:demand_progress}
If there exists a demand spot not occupied, a new bubble is eventually created.
\end{lemma}
\begin{proof}
There is a demand spot that is not occupied.
There is a demand root~$d$ that is responsible for this demand spot.
If $d$ currently holds a bubble, it eventually contracts by Lemma~\ref{lem:bubble_progress}.
After $d$ is contracted, next time it activates it expands, thus creating a new bubble.
\end{proof}

The lemmas outlined above lead to the following conclusion.
\begin{lemma}\label{lem:correctness}
A particle configuration $\PC$ stays valid and live at all times.
\end{lemma}

\subsection{Running time}
We begin by arguing that the total distance traveled by each bubble is linear in the number of particles in the system.
We first limit the length of the path each bubble takes.

\begin{lemma}\label{lem:edge_counts}
The total path of a bubble~$b$ in a particle configuration $\PC$ with $n$ particles has size $O(n)$.
\end{lemma}
\begin{proof}
	For ease of analysis, we say each bubble $b$ initially has a budget $\mathcal{B}_0$ of size equal to the number of edges in $S$.
	Whenever $b$ makes a step towards supply, the budget decreases by one.
	Whenever $b$ makes a step towards demand, the budget increases by one.
	A bubble only takes a step back towards a demand root if none of the valid successors in the supply graph are alive.
	Therefore, it only moves back over every edge of $S$ at most once and the budget is increased in total by $O(n)$.
	The budget of a bubble $b$, and therefore the number of steps $b$ takes is bounded by $O(n)$.
	
	We now prove that the budget is sufficient for $b$ to eventually reach a supply.
	Let $b$ be initiated at demand root $d$, and let $w$ be the first node at which $b$ changes its direction from $\delta=\textsc{s}$ to $\textsc{d}$, and let $v$ be the first node at which $b$ changes its direction back to $\delta=\textsc{s}$.
	Note that $w\not=v$, as $b$ always moves down or up $S$ along a feather path.
	Thus, if the direction of a bubble changes, it makes at least one step up the supply graph.
	The budget of $b$ when it reaches $w$ is $\mathcal{B}_0-\dist(d,w)$ (where $\dist(d,w)$ is the distance between $d$ and $w$ in $S$), and its budget when it reaches $v$ is $\mathcal{B}_v=\mathcal{B}_0-\dist(d,w)+\dist(w,v)$.
	By the triangle inequality ($\dist(d,v)+\dist(v,w)\ge\dist(d,w)$), $\mathcal{B}_v$ is at least the budget that $b$ would have if it traveled to $v$ directly from $d$.
	Repeating this argument between any two consecutive changes of $b$'s direction from $\textsc{d}$ to $\textsc{s}$, we conclude that at any moment in time $b$'s budget is sufficient to reach any supply root along a feather path.
\end{proof}

Next, we analyze the progress made by each bubble by creating a specific series of configurations; we show that any activation sequence results in at least as much progress.
This approach has been previously used for analyzing the running time of a moving line of particles~\cite{Derakhshandeh2016}, and for a tree moving from the root~\cite{Daymude2018}.
We repeat most of the definitions of~\cite{Daymude2018} here for completeness.
We say a \emph{parallel movement schedule} is a sequence of particle configurations $(C_0,\dots,C_t)$ where each configuration~$C_i$ is a valid particle configuration, and for each $i \geq 0$, $C_{i+1}$ can be reached from $C_i$ such that for every particle $p$, one of the following properties hold:
\begin{enumerate}
    \item $p$ occupies the exact same node(s) of $G$ in $C_i$ and $C_{i+1}$,
    \item $p$ was contracted in $C_i$ and expands into an empty adjacent node in $C_{i+1}$,
    \item $p$ was expanded in $C_i$ and contracts, leaving an empty node in $C_{i+1}$, or
    \item $p$ performs a valid handover with a neighboring particle $q$.
\end{enumerate}

Let $A$ be an arbitrary fair asynchronous activation sequence and let $C_i^A$ be the particle configuration at the end of round $i$ in $A$, if every particle moves according to our algorithm.
Let $\pi_p$ be the path particle~$p$ takes from $C_0^A$ till $C_k^A$, where $C_k^A$ is the configuration that solves the reconfiguration problem.

A \emph{bubble schedule} $B = (A, (C_0, \dots, C_t))$ is a parallel schedule $(C_0, \dots, C_t)$, where each particle $p$ follows the same path $\pi_p$ it would have followed according to $A$.
Recall that if two particles $p$ and $q$ want to take a conflicting action by pulling on the same neighbor, this is arbitrarily resolved and only one of them succeeds.
In a bubble schedule the particles follow the same path as they would under $A$, and hence the same particle makes progress in the bubble schedule.
A bubble schedule is said to be \emph{maximal} if all particles perform the allowed movements in a parallel schedule whenever possible.
If all particles follow the same paths, it follows that the bubbles follow the same paths as well.

Let $M^p = (m_0^p, \dots, m_j^p)$ be the sequence of moves that a particle $p$ makes on $\pi_p$.
We construct the configurations in a maximal bubble schedule starting from $C_0 = C_0^A$ as follows.
At each configuration $C_i$ for $0 \leq i < t$, let $\{m_{i}^p \mid p \in P\}$ be the set of moves all particles can make at configuration $C_i$.
We now pick a maximal subset $M_i$ of $\{m_{i}^p \mid p \in P\}$ such that all moves in $M_i$ are mutually compatible.
After executing the moves in $M_i$, we obtain~$C_{i+1}$.
Lastly, we say that a configuration $C$ is dominating $C'$ ($C \succeq C'$) if every particle in $C$ has traveled at least as far along its path as its counterpart in $C'$.

\begin{figure*}
    \centering
    \includegraphics{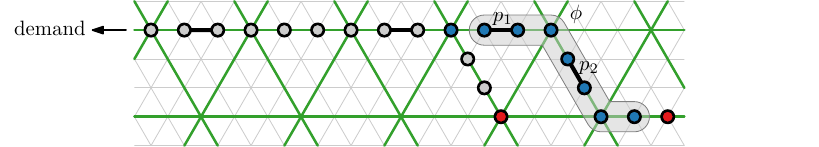}
    \caption{A partitioning of $S$. The grid $G_L$ is shown in green. Only particles on $S$ are shown. Demand root is on the left outside the figure, supply roots are red. Particles $p_1$ and $p_2$ are adjacent.
    The blue particles are on nodes of $V_c$, the grey particles are on nodes of $V_b$. One component~$\phi$ is highlighted in light grey.}
    \label{fig:partitioning}
\end{figure*}

\begin{lemma}[Lemma~3 of~\cite{Daymude2018}]
\label{lem:domination_argument}
Given any fair asynchronous activation sequence $A$ which begins at an initial configuration $C_0^A$ in which there are no expanded particles, there is a maximal bubble schedule $B = (A, (C_0, \dots, C_t))$ with $C_0 = C_0^A$, such that $C_i^A \succeq C_i$ for all~$0 \leq i \leq t$.
\end{lemma}
Any valid activation sequence $A$ makes at least as much progress per round as a corresponding maximal bubble schedule $B$.
Now for a maximal bubble schedule $B$, consider a similar bubble schedule~$B'$, where contracted particles at the grid nodes of $G_L$ are only allowed to activate if there is no other grid node in $G_L$ that is occupied by an expanded bubble.
Clearly $B$ (and thus $A$) makes at least as much progress per configuration as $B'$.
We call such $B'$ a \emph{junction-synchronized bubble schedule}, and use this kind of parallel schedule from now on.

%
Consider a configuration $C_i$, where all expanded particles in $S$ occupy only edge nodes of $G_L$.
We call two expanded particles \emph{adjacent} if they are separated by a single grid node.
Let $J$ be the set of all junction nodes of $S$ with total (in and out) degree at least three, and let $R_D$ and $R_S$ be the sets of demand and supply roots respectively.
We partition the nodes of $S$ into two sets $V_b$ and $V_c$, such that (1) each connected component $\phi$ of $S\setminus (J\cup R_D \cup R_S)$ (which is a path) consists of $k_b\ge 0$ nodes in $V_b$ followed by some number $k_c$ nodes in $V_c$, and (2) if $k_c\ge 5$ then the first five nodes on $\phi$ that are in $V_c$ contain two expanded particles separated by a grid node (see Figure~\ref{fig:partitioning}).

If these two conditions are satisfied, we call such partition $(V_b,V_c)$ \emph{valid}.
Consider the particles occupying the nodes of $S\setminus (J\cup R_D \cup R_S)$.
Recall that for a bubble~$b$ moving forward in $S$, $N(b)=\{(r^i_1,r^i_2) \mid i\in \mathcal{D}(b)\}$ is the set of pairs of edge nodes lying in the valid directions for $b$.
For every expanded particle $p$ with bubble~$b$, we put the nodes of $p$ in $V_c$ if there does not exists $(r_1, r_2)\in N(b)$ where both $r_1$ and $r_2$ are occupied by contracted particles.
Next, for every unassigned particle $p \in S\setminus (J\cup R_D \cup R_S)$ (expanded or not), consider the component $\phi$ that contains $p$.
We put the nodes of $p$ in $V_c$ if there exists an expanded particle $q$ in $\phi$ such that (1) $q$ has $p$ as its predecessor in $\phi$, and (2) $q$ has been assigned to $V_c$.
All the remaining nodes in $S\setminus (J\cup R_D \cup R_S)$ are put in $V_b$.
Finally, for each node $c$ in $J \cup R_D \cup R_S$, if it has an adjacent node in $V_c$, we put $c$ in $V_c$, otherwise, we put $c$ in $V_b$.
Note that by construction this partition is valid.
\begin{lemma}\label{lem:nice-schedule}
Given any fair asynchronous activation sequence $A$, consider any junction-synchronized bubble schedule $(A, (C_0,\cdots,C_t))$.
For any configuration $C_i$ with $0 \le i \le t$ that has all expanded particle in $S$ occupying two edge nodes in $G_L$, there is a valid partition of the nodes of $S$ into $V_b$ and $V_c$ that satisfies the following properties.
\begin{enumerate}
\item For induced subgraph $S_b$ of $S$ on $V_b$: For every bubble~$b\in S_b$, there exists $(r_1, r_2)\in N(b)$ where both $r_1$ and $r_2$ are in $V_b$ and are occupied by contracted particles.
\item For induced subgraph $S_c$ of $S$ on $V_c$: For every two neighboring edge nodes $r_1$ and $r_2$ in $G_L$ occupied by two contracted particles, there exists a bubble~$b$ in $S_c$ such that $(r_1, r_2)\in N(b)$.
\end{enumerate}
\end{lemma}
\begin{proof}
	We will prove the statement by induction.
	Because there are no bubbles in $C_0$, all the nodes of $S$ are in $V_b$ and the lemma holds.
	
	Assume the lemma holds for some configuration $C_i$ with $0 \le i \le t$ in which every expanded particle~$p$ occupies only edge nodes in $G_L$.
	We will now show that the lemma also holds for configuration $C_{i+x}$ for some small $x$ and we show this is the first configuration for which every expanded particle~$p$ again occupies only edge nodes in $G_L$.
	
	If there are no bubbles in $S$ in $C_i$, the demand roots expand into their demand components, creating new bubbles.
	New bubbles that are created can only enter $S$ if there is enough free space after the demand root, therefore, a new bubble automatically satisfies the lemma.
	
	Assume then that there are bubbles in $S$ in $C_i$.
	For $C_i$ there can be two cases: there are either none or some expanded particles which received an acknowledgment to move from the corresponding junction nodes in the previous round $C_{i-1}$.
	In the first case, no particles move in $C_i$, and only request, availability, and acknowledgment tokens are being sent.
	Then, the lemma still holds for $C_{i+1}$.
	In the second case, particles that have received an acknowledgement move in $C_i$.
	It takes three more configurations, until $C_{i+3}$, for them to fully traverse their junction nodes and occupy the following edge nodes.
	By our assumption on the junction-synchronized bubble schedule, no other particles move in rounds $C_{i+1}$ and $C_{i+2}$.
	
	We first analyze the subgraph $S_b$, induced by the nodes in $V_b$ at $C_{i+3}$.
	Consider an expanded particle $p$ in $S_b$ with bubble $b$, for which the first property is violated in $C_{i+3}$.
	This can only happen if either (1) $p$ does not move in $C_i$, $c(b)\in J$, and another bubble crosses $c(b)$ in rounds $C_i$--$C_{i+3}$, or (2) $p$ moves in~$C_i$, and all nodes of $N(b)$ in $C_{i+3}$ belong to $V_c$.
	Then, we change the partition of the nodes in $C_{i+3}$ by moving the nodes of $p$ and $c(b)$ from $V_b$ to $V_c$.
	This does not violate the second property for $p$ in $C_{i+3}$, and the partition remains valid.
	
	Now, we analyze the subgraph $S_c$.
	Consider two contracted particles occupying edge nodes $r_1$ and $r_2$ in $S_c$, for which the second property is violated in $C_{i+3}$.
	Let $r_3$ and $r_4$ in $V_c$ be a pair of grid nodes which together with $r_1$ and $r_2$ violate the second property in~$C_{i+3}$, and $c_0$, $c_1$, and $c_2$ be adjacent to them grid nodes such that  $c_2 \rightarrow r_3\rightarrow r_4\rightarrow c_1\rightarrow r_1\rightarrow r_2\rightarrow c_0$ (see Figure~\ref{fig:vc_order}).
	
	\begin{figure}
		\centering
		\includegraphics{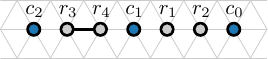}
		\caption{The order of particles at $C_i$. All particles are on $S$. The particles on crossings are blue. Supply is to the right, demand to the left.}
		\label{fig:vc_order}
	\end{figure}
	
	If the particles at $r_1$ and $r_2$ do not move in $C_i$, then $r_3$ and $r_4$ are occupied by an expanded particle in $C_i$, which moves to a different branch of $S$.
	Then it must be that $c_1\in J$, and thus the partition $(V_b,V_c)$ is not valid in $C_i$.
	Consider now the case that there is an expanded particle at $r_1$ and $r_2$ in $C_i$ that moved.
	Then, the bubble $b$ for which $r_3$ and $r_4$ are in $N(b)$ in $C_i$, crosses $c_2$ into a different branch of $S$ in $C_i$ to $C_{i+3}$.
	We change the partition of the nodes in~$C_{i+3}$ by moving the nodes of $r_1$, $r_2$, $r_3$, and $r_4$ from $V_c$ to $V_b$.
	If $c_0$ and $c_1$ do not have other adjacent nodes in $V_c$, then we move them into $V_b$ as well.
	This does not violate the first property for $p$ in $C_{i+3}$, and the partition $(V_b,V_c)$ remains valid.
\end{proof}
\medskip
This lemma ensures that in a junction-synchronized bubble schedule, bubbles in $S_b$ have enough space to move, and for expanded particles to pack tighter towards the sinks of $S_c$, in every constant number of rounds.
As a result, a bubble will cross a sink of $S_c$ and enter $S_b$ (or supply) every constant number of rounds.
Therefore, it takes a bubble in $S_c$ at most linear, in the size of $S_c$, number of rounds to exit $S_c$.
Finally, bubbles that move towards an empty supply will reverse to the nearest alive path to find non-empty supply.
This in total takes linear number of rounds in the size of the dead path.
Together with Lemma~\ref{lem:edge_counts}, these observations imply the following lemma.

\begin{lemma}\label{lem:maximal_linear}
A junction-synchronized bubble schedule $(A,(C_0,\cdots,C_t))$ solves the reconfiguration problem in $O(n)$ configurations.
\end{lemma}
\begin{proof}
We know that the path of each bubble is linear (Lemma~\ref{lem:edge_counts}).
What is left to show is that every bubble makes progress every constant number of configurations.
The analysis of~\cite{Daymude2018} directly applies to the movement of bubbles inside the supply and demand components.
Thus, we only need to analyze the movement of bubbles through $S$.
Consider a bubble $b$ and its path $\pi_b$ in $S$ starting at some demand root in configuration $C_0$ and reaching the supply root in some configuration $C_{t_b}$.
Let $v_i(b)$ denote the node of $S$ that $b$ occupies in configuration $C_i$.
If $v_i(b)$ is in $V_b$ in $C_{i}$, then, by Lemma~\ref{lem:nice-schedule}, $b$ made progress along $\pi_b$ since at least the last three rounds.

If $v_i(b)$ is in $V_c$ in $C_{i}$ and later, at some configuration $C_{i+j}$, it crosses some sink of $S_c$ back to $S_b$, then it makes on average $O(j)$ progress along its path.
This is due to the fact that a bubble crosses a sink of $S_c$ at most every three rounds.

Finally, if $v_i(b)$ is in $V_c$ in $C_{i}$ and later at some configuration $C_{i+j}$ its node $v_{i+j}(b)$ becomes dead, 
then $b$ reverses its direction, and this branch of $S$ is not going to be traversed by other bubbles in the future.
After a linear number of rounds in the size of the dead branch, all the bubbles on it reverse their directions, and the analogous analysis as in the previous case applies.
As the total size of the dead branches is limited by $O(n)$, the number of rounds that each bubble spends to traverse dead branches is at most $O(n)$.

Thus, in total, every bubble will reach supply and get resolved in $O(n)$ rounds.
\end{proof}

Together with the proof of correctness (Lemma~\ref{lem:correctness}), we conclude with the following theorem:
\begin{theorem}\label{theorem:sequential}
The particle reconfiguration problem for particle configuration $\PC$ with $n$ particles can be solved using $O(n)$ rounds of activation under a sequential scheduler.
\end{theorem}

\subparagraph{Extending the analysis to an asynchronous scheduler}
Unlike in the case of a sequential scheduler, when at each moment only one particle can be active, under an asynchronous scheduler multiple particles may be active at the same time.
This may lead to concurrency issues.
Specifically, when actions of two or more particles conflict with each other, not all  of them can be finished successfully.
To extend our algorithm to the case of an asynchronous scheduler, we explore what kind of conflicts may arise, and ensure that our algorithm can deal with them. 
There can be three types of conflicts:
\begin{enumerate}
\item Two particles try to expand into the same empty node.
\item Two contracted particles try to push the same expanded particle.
\item Two expanded particles try to pull on the same contracted particle.
\end{enumerate}

The first two kinds of conflicts never arise in our algorithm, as (1) every empty node (demand spot) can be moved into from only a single direction, according to the spanning tree stored by the corresponding demand root, and (2) no particle ever performs a push operation.
The third conflict could potentially arise in our approach when two bubbles traveling on crossing paths pass the junction node.
However, as crossing of a junction is controlled by the junction particle, only one particle at a time is given permission to pull on the junction node.
Thus, under an asynchronous scheduler, all actions initiated by the particles always succeed.

We are left to analyze whether the asynchronous setting can lead to deadlocks in our algorithm, where concurrency issues result in some particle not being able to make progress.
As mentioned above, crossings of junctions are controlled by the junction particles, and thus cannot lead to deadlocks.
Furthermore, the algorithm forbids the simultaneous existence of bidirectional edges in $S$, thus there cannot be deadlocked bubbles moving in the opposite directions over the same edge.
The only remaining case of a potential deadlock is when the algorithm temporarily blocks one of the two directions of a bidirectional edge.
In a sequential schedule, this choice is made by one of the two edge nodes of $G_L$ corresponding to the bidirectional edge of $S$.
However, in an asynchronous schedule, both edge nodes may become active simultaneously, and choose the opposite directions of $S$.
To resolve this case, and to break the symmetry of two nodes activating simultaneously and choosing the opposite directions, we need to utilize the power of randomness.
We assume that the particles have access to a constant number of random bits.
Daymude et al.~\cite{Daymude2021a} show that, in that case, a mechanism exists that allows the particles to lock their local neighborhood, and to perform their actions as if the neighboring particles were inactive.
Such a locking mechanism increases the running time only by a constant factor in expectation, and by a logarithmic factor if the particle needs to succeed with high probability.
Thus, with a locking mechanism, we can ensure that our algorithm can select one of the two directions of a bidirectional edge in $S$.
\begin{theorem}\label{theorem:asynchronous}
The particle reconfiguration problem for particle configuration $\PC$ with $n$ particles can be solved in $O(n)$ rounds of activation in expectation (or in $O(n\log n)$ rounds of activation with high probability) under an asynchronous scheduler, if each particle has access to a constant number of random bits.
\end{theorem}

The running time of our algorithm is worst-case optimal: any algorithm to solve the particle reconfiguration problem needs at least a linear number of rounds.

\begin{theorem}
Any algorithm that successfully solves the particle reconfiguration problem needs $\Omega(n)$ rounds.
\end{theorem}
\begin{proof}
Let the initial shape be a horizontal line of $n$ particles.
Let the target shape be a horizontal line of $n$ particles such that the rightmost particle of the initial shape is the leftmost position of the target shape.
Now the leftmost particle of the initial shape needs to travel at least $n-1$ nodes.
Because a particle executes at most one move operation per activation, a particle needs at least $2$ rounds for every node it moves.
Therefore, this particle needs at least $2n-2$ rounds to reach the target shape.
Any algorithm that solves this problem needs $\Omega(n)$ rounds.
\end{proof}

\section{Non-simply Connected}\label{sec:holes}
Our algorithm assumes that the core of $I\cap T$ is simply connected (assumption \assumption{1}).
If the core consists of more than one component, then we can choose one component as the core for the reconstruction.
The other components can then be interpreted by the algorithm as both supply and demand.
Consequently they will first be deconstructed and then reassembled.
Clearly this is not an ideal solution, but it does not affect the asymptotic running time and does not require any major changes to our algorithm.

If the core is not simply connected, that is, it contains holes, then our procedure for creating feather trees may no longer create shortest path trees or it might not even terminate.
Synchronizing the trees around holes is difficult while creating multiple trees at the same time.
However, creating a single shortest path tree in an arbitrary connected configuration (including those with holes) has been done before~\cite{boulinier2008space} (Lemma~\ref{lem:boulinier_sp}).
In a preprocessing step we can \emph{cut} a configuration in such a way that the resulting configuration is simply connected, while increasing the maximum depth of any feather tree by at most a factor two.
The algorithm described in~\cite{boulinier2008space} takes a leader particle~$\ell$, and computes for each particle if its neighbors are closer, further, or equal distance from $\ell$.
Note that every particle~$p$ has at least one neighbor that is closer to $\ell$ than $p$.
We use the information provided by this algorithm to determine where to cut a configuration~$\PC$ such that the underlying graph $G_\PC$ becomes simply connected.
With slight abuse of notation, we say that a particle $p\in\PC$ occupies node $p$ of $G_\PC$.

Consider two different paths $\pi_1$ and $\pi_2$ having a common node~$s$ as start, and a common target node~$t$, or adjacent target nodes~$t_1$ and $t_2$.
Paths $\pi_1$ and $\pi_2$ are \emph{homologically identical} if all the nodes of $G$ that are inside the area $A$ enclosed by $\pi_1$, $\pi_2$ (and edge $(t_1, t_2)$) are also part of $G_\PC$, i.e. the area contains no holes.
The paths are homologically distinct otherwise.
We can use this notion of distinct paths to detect if a configuration is simply connected.

\begin{lemma}\label{lem:no-holes}
    Let $\PC$ be a configuration with a leader particle $\ell$. If for every particle $p$ and every edge $(p, p')$ all shortest paths to $\ell$ are homologically identical, then $\PC$ is simply connected.
\end{lemma}
\begin{proof}
    Assume for contradiction that $\PC$ contains a hole~$h$, but that for all particles $p$ and all edges $(p, p')$, all shortest paths to $\ell$ are homologically identical.
    Let $H$ be the geodesic convex hull of $\ell$ and $h$.
    Now let $m$ be the point on $h$ such that the distance from $m$ to $\ell$ on $h$ is the same both clockwise as well as counter-clockwise.
    Because $h$ is the geodesic convex hull, $m$ lies either on a particle, or in the middle of an edge.
    Moreover, both paths from $\ell$ to $h$ are shortest paths and homologically distinct.
    This is a contradiction and hence $\PC$ cannot contain a hole.
\end{proof}

Since our goal is to make the configuration simply connected, it suffices to make a configuration such that all shortest paths from a particle~$p$ to a leader~$\ell$ are homologically identical.
Therefore, particles that have at least two homologically different shortest paths to $\ell$ are of special interest.
We say a particle~$p$ lies on a \emph{bisector} from leader $\ell$ if there are at least two homologically distinct shortest paths from $\ell$ to $p$.
An edge $(p, p')$ of $G_\PC$ lies on a bisector from leader $\ell$ if all shortest paths from $\ell$ to $p$ are homologically distinct from all shortest paths from $\ell$ to $p'$, see Figure~\ref{fig:bisector_example}.

We will now show how to locally detect some of these particles and edges. For a particle~$p$, let $d_p$ be the distance from the leader $\ell$ to $p$, and let $N_{-1}$ ($N_{+1}$) be all particles neighboring $p$ that are closer to (further from) $\ell$, i.e. $d_p > d_v$ ($d_p < d_v$) for all $v\in N_{-1}$ ($v\in N_{+1}$).

\begin{figure}[t]
    \centering
    \includegraphics{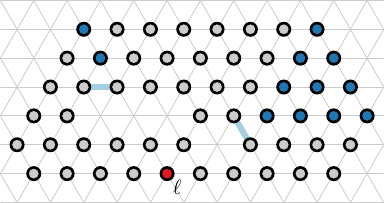}
    \caption{A configuration with holes. Leader $\ell$ is marked red. The blue particles and edges lie on a bisector from $\ell$}
    \label{fig:bisector_example}
\end{figure}

\begin{lemma}\label{lem:node_in_A_has_path}
    Let paths $\pi_1$ and $\pi_2$ from $s$ to $t$ be two homologically identical shortest paths enclosing an area $A$. Then for every node $n\in A$, there exists a shortest path from $s$ to $t$ via $n$.
\end{lemma}
\begin{proof}
    In the case that $A$ consists of multiple disconnected components, i.e. $\pi_1$ and $\pi_2$ overlap at a particle that is not $s$ nor $t$, we restrict ourselves to one of these components.
    Redefine $s$ to be the last particle before $A$ where $\pi_1$ and $\pi_2$ overlap, and redefine $t$ to be the first particle after $A$ where $\pi_1$ and $\pi_2$ overlap again.
    Since both $\pi_1$ and $\pi_2$ are shortest paths, their segments between $s$ and $t$ are still equal in length.
    Because $\pi_1$ and $\pi_2$ now do not overlap, but are equal in length, they must lie completely inside the visibility region of $s$, since otherwise they would meet at some root of the SPM according to Corollary~\ref{the:spm}. Hence, they must both be $60^\circ$-angle monotone, using the same $60^\circ$-cone with fixed orientation $o$.
    Now there exists a $60^\circ$-angle monotone path $\pi'_1$ from $s$ to $n$ as well as a $60^\circ$-angle monotone path $\pi'_2$ from $n$ to $t$ both using the same cone with orientation $o$.
    The paths combining $\pi'_1$ and $\pi'_2$ is a shortest path from $s$ to $t$ via $n$.
\end{proof}

\begin{lemma}\label{lem:edge_detection}
    Let $(p, p')$ be an edge of $G_\PC$ with a leader $\ell$ such that $d_p = d_{p'}$, and let $n_1$ and $n_2$ be the common neighbors of $p$ and $p'$. If and only if $d_{n_1} \geq d_p$ and $d_{n_2} \geq d_p$, then $(p, p')$ is on a bisector with respect to $\ell$.
\end{lemma}
\begin{proof}
    Assume that $d_{n_1} \geq d_p$ and $d_{n_2} \geq d_p$.
    Now assume for contradiction that there exist paths $\pi_1$ from $\ell$ to $p$ and $\pi_2$ from $\ell$ to $p'$ such that all nodes in the area~$A$ enclosed by $\pi_1$ and $\pi_2$ are contained in $G_\PC$.
    Since $d_p = d_{p'}$ and all nodes in $A$ are present, either $n_1$ or $n_2$ lies in $A$.
    Without loss of generality, let $n_1$ be within $A$.
    Due to Lemma~\ref{lem:node_in_A_has_path}, there exists a shortest path via $n_1$ to both $p$ and $p'$. Hence, $n_1$ is closer to $\ell$, i.e. $d_{n_1} < d_p$. This is a contradiction and hence if $d_{n_1} \geq d_p$ and $d_{n_2} \geq d_p$, then $(p, p')$ is on a bisector with respect to $\ell$.

    On the other hand, assume for contradiction that $(p, p')$ is on a bisector, but (without loss of generality) $d_{n_1} < d_p$.
    Then there must exist a shortest path from $\ell$ to $p$ via $n_1$.
    Hence, there also exists a shortest path from $\ell$ to $p'$ via $n_1$, that only differs in the last edge.
    This means that there are two shortest paths towards $p$ and $p'$ that are homologically identical and $(p, p')$ would not be on a bisector by definition. This is a contradiction.
\end{proof}

\begin{lemma}\label{lem:node_detection}
    If $G_{N_{-1}}$ is not connected, then $p$ lies on a bisector.
\end{lemma}
\begin{proof}
    Let $N_1$ and $N_2$ be two of the connected components of $G_{N_{-1}}$.
    Take two shortest paths $\pi_1$ and $\pi_2$ that go via $n_1\in N_1$ and $n_2\in N_2$ respectively, and let $A$ be the area enclosed by $\pi_1$ and $\pi_2$.
    Because $N_1$ is not connected to $N_2$, there must be a neighbor $n_3\not\in N_{-1}$ that lies inside of $A$, but with $d_{n_3} > d_p$.
    If $\pi_1$ and $\pi_2$ would be homologically identical, then by Lemma~\ref{lem:node_in_A_has_path}, there must be a shortest path via $n_3$ to $p$.
    This contradicts $d_{n_3} > d_p$ and hence $\pi_1$ and $\pi_2$ are homologically distinct and $p$ lies on a bisector.
\end{proof}

Due to Lemma~\ref{lem:edge_detection} we can detect any edge on a bisector by examining its two common neighbors, and because of Lemma~\ref{lem:node_detection} we can detect some particles on a bisector by looking at their neighbors.

We use this information as follows to cut the shape into a simply connected configuration.
Whenever a particle \emph{cuts} an edge, it sends this information to its neighbor. Note that the cut edge can still be used to pass information for the purpose of the cutting algorithm.

Every particle will have one of three states: \emph{start}, \emph{primed}, and \emph{finished}.
A particle can also be \emph{marked} as being on a bisector.
Initially, all particles will be in the \emph{start} state.
The leader~$\ell$ puts itself into the \emph{primed} state.
Then, whenever a particle~$p$ activates, it puts itself in the \emph{primed} state if all particles $p'\in N_{-1}$ are also \emph{primed}.
If a particle \emph{primes} itself, it first guarantees a path towards $\ell$ exists and keeps existing.
It does this by picking a neighbor~$p'\in N_{-1}$ for which the edge $(p, p')$ is not cut yet and marking the edge $(p, p')$ as \emph{not-cut}.
If no such $p'$ exists, then it picks an arbitrary $p'\in N_{-1}$, repairs the edge $(p, p')$ and then marks it as non-cut.
Non-cut edges are henceforth never cut anymore.
If a particle \emph{primes} itself, it checks if it is on a bisector as follows.
If the particle~$p$ is part of an edge $(p, p')$ that is on a bisector due to Lemma~\ref{lem:edge_detection}, the respective edge is cut and both common neighbors are marked as being on the bisector as long as they are not \emph{primed} yet.

If the particle~$p$ is a node on a bisector due to Lemma~\ref{lem:node_detection}, or it was marked as being on a bisector by a neighbor before priming itself, it cuts every single edge except the single non-cut edge.
For each connected component of $G_{N_{+1}}$ it chooses one particle and marks it as being on a bisector.

Whenever for a \emph{primed} particle~$p$, all neighboring particles $p'\in N_{+1}$ are \emph{finished}, it marks itself as \emph{finished}. This also happens if $N_{+1}$ is empty.
Finally, the algorithm is done when the root~$\ell$ is \emph{finished}. The algorithm can be seen in Algorithm~\ref{alg:cut}, and an example of the final cut can be seen in Figure~\ref{fig:cutting}.

\begin{algorithm*}
    \caption{Algorithm for cutting the core, ran on a particle $p$}\label{alg:cut}
    \begin{algorithmic}
        \If{$p$.state = \emph{start}}
            \If{$p = \ell \textbf{ or } \forall p'\in N_{-1}$: $p'$.state = \emph{primed}}
                \State $p$.state $\leftarrow$ \emph{primed}
                \State Pick a $p'\in N_{-1}$ such that $(p, p')$ is not cut yet. Mark $(p, p')$ as non-cut.
                \State If no such $p'$ exists, pick a $p'\in N_{-1}$, repair $(p, p')$ and mark it as non-cut.
                \If{edge on a bisector with $p''$ (Lemma~\ref{lem:edge_detection}) and $p''$.state = \emph{start}}
                    \State Cut edge $(p, p'')$
                    \State Mark the common neighbors of $(p, p'')$ as nodes on the bisector
                \EndIf
                \If{marked as node on a bisector (by another node or Lemma~\ref{lem:node_detection})}
                    \State Cut every single edge, except for the non-cut edge
                    \State For each connected component of $G_{N_{+1}}$, mark one node on a bisector.
                \EndIf
            \EndIf
        \EndIf
        \State
        \If{$p$.state = \emph{primed}}
            \If{$\forall p'\in N_{+1}$: $p'$.state = \emph{finished}}

                \State $p$.state $\leftarrow$ \emph{finished}
            \EndIf
        \EndIf
        \State
        \If{$p$.state = \emph{finished}}
            \State Do nothing
        \EndIf
    \end{algorithmic}
\end{algorithm*}

\begin{figure}[t]
    \centering
    \includegraphics{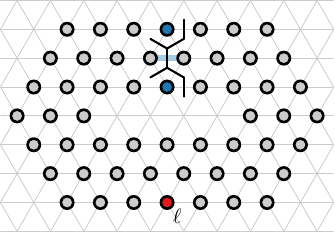}
    \caption{A configuration with a hole. Leader $\ell$ is marked red. The bottom blue particle is marked at the start, since it has two independent neighbors closer to $\ell$. The blue edge is also marked at the start. The upper blue particle is marked during the execution of the algorithm. The black lines represent the final cuts in $G_\PC$.}
    \label{fig:cutting}
\end{figure}

\begin{lemma}\label{lem:connected_after_cut}
    The configuration is connected after the cut, and distances to $\ell$ are the same as before the cut.
\end{lemma}
\begin{proof}
    Every single particle~$p$ is still connected to at least one neighbor that is on a shortest path to $\ell$ before the cut, since every particle denotes one edge as non-cut.
    Therefore, every single node has at least one shortest path towards the original leader $\ell$, that was also a shortest path before the cut.
\end{proof}

\begin{lemma}
    The configuration is simply connected after the cuts.
\end{lemma}
\begin{proof}
    We show that there is no node or edge that can be reached by two homologically different shortest paths to $\ell$, which is sufficient according to Lemma~\ref{lem:no-holes}.
    First, observe that every particle goes from \emph{start} to \emph{primed}, as well as from \emph{primed} to \emph{finished} exactly once.
    Any edge whose nodes are reached by two homologically distinct paths is cut as soon as its respective particles go from \emph{start} to \emph{primed}.
    After the cut all edges are reached by homologically identical shortest paths.
    Now let $p$ be a node that can be reached by two homologically distinct shortest paths before the algorithm.
    We will show that after the cut, $p$ does not have two distinct shortest paths to $\ell$ anymore.

    If a node~$p$ can be reached by two homologically distinct paths $\pi_1$ and $\pi_2$, then when walking back over $\pi_1$ and $\pi_2$ towards $\ell$, there must be a node where they diverge.
    We will recursively walk back over $\pi_1$ and $\pi_2$ until we find this node, and show that a cut has been made at that point.
    
    If a node~$p$ can be reached by two homologically distinct paths $\pi_1$ and $\pi_2$, there are two options.
    Either a neighbor $p'\in N_{-1}$ lies on both $\pi_1$ and $\pi_2$, or no such neighbor exists.    
    If such a neighbor $p'$ exists, we recursively analyze this neighbor.
    Otherwise, $\pi_1$ and $\pi_2$ split here. 
    We again have two cases.
    Let $p'_1$ and $p'_2$ be the two neighbors on paths $\pi_1$ and $\pi_2$ respectively.
    Nodes $p'_1$ and $p'_2$ can either be neighbors of each other or not.
    
    \subparagraph{Neighbors}
    In the first case, we claim that $(p'_1, p'_2)$ is an edge on the bisector.
    To show this, let $p'_3$ be the common neighbor of $p'_1$ and $p'_2$ that is not $p$.
    If $p'_3$ is closer to $\ell$ than $p'_1$ and $p'_2$, then $\pi_2$ can be changed to go via $p'_1$, contradicting the fact that there does not exist a single neighbor $p'$ on both $\pi_1$ and $\pi_2$.
    Therefore, $p'_3$ and $p$ are both at least as far away from $\ell$ as $p'_1$ and $p'_2$, and hence according to Lemma~\ref{lem:edge_detection}, $(p'_1, p'_2)$ is an edge on a bisector.
    This means that $p$ would be marked as being on the bisector during the algorithm and has cut all edges except for one.
    Hence, after the cut, it cannot be reached by both $\pi_1$ and $\pi_2$.

    \subparagraph{Not neighbors}
    In the second case, $p'_1$ and $p'_2$ are not neighbors.
    Now $G_{N_{-1}}$ cannot be disconnected, since otherwise $p$ would have been marked at the start as being on a bisector.
    Therefore, let $G_{N_{-1}}$ be connected, and assume that around $p$, between $p'_1$ and $p'_2$ there exists only a single node $p'_3$, which because $G_{N_{-1}}$ is connected, has $d_{p'_3} = d_p - 1$.
    Now both $(p'_1, p'_3)$ and $(p'_2, p'_3)$ cannot be edges on a bisector, since otherwise $p$ would have been marked.
    This means that both $\pi_1$ and $\pi_2$ could have gone via $p'_3$.
    This contradicts our initial assumption that there is no $p'\in N_{-1}$ reachable by the same two distinct paths as $p$ and the case that $p'_1$ and $p'_2$ are not neighbors does not occur.

    In all cases, by recursively following $\pi_1$ and $\pi_2$, we either find a node that has been marked by the algorithm, or we find $\ell$.
    However, if we find $\ell$, then $\pi_1$ is not distinct from $\pi_2$, which is a contradiction.
    Hence, by recursively following $\pi_1$ and $\pi_2$, we will find a node that has been marked by the algorithm and at that place the appropriate cut has been made.
    Therefore, after the cuts, there does not exist a node or edge that is reachable by two homologically distinct shortest paths and the configuration is simply connected
\end{proof}

We can use this procedure as a pre-processing step to make the core of our reconfiguration problem simply connected. We can use an arbitrary leader particle~$\ell$.

\begin{lemma}\label{lem:twice_diameter}
    The diameter of the simply connected configuration after the cuts is at most twice the diameter of the original configuration.
\end{lemma}
\begin{proof}
    Due to Lemma~\ref{lem:connected_after_cut}, every particle has a shortest path to the root $\ell$ that does not change after the cuts. This distance is at most the diameter of the original configuration.
    Hence, after the cuts, every two particles $s$ and $t$ are connected via $\ell$ and the length of this path is at most twice the original diameter.
\end{proof}

The algorithm to calculate distances from $\ell$ to each particle takes $O(d)$ rounds for a configuration $\PC$ with diameter $d$.
The cutting pre-processing step itself also takes $O(d)$ rounds.
Combining this with Lemma~\ref{lem:twice_diameter} shows that the asymptotic running time of our reconfiguration algorithm is not changed when applying this preprocessing step, lifting our original assumption.

\section{Discussion}\label{sec:conclusion}

We presented the first reconfiguration algorithm for programmable matter that does not use a canonical intermediate configuration. In the worst case, our algorithm requires a linear number of activation rounds and hence is as fast as existing algorithms. However, in practice, our algorithm can exploit the geometry of the symmetric difference between input and output and can create as many parallel shortest paths as the problem instance allows. 

We implemented our algorithm in the Amoebot Simulator\footnote{\url{https://github.com/SOPSLab/AmoebotSim}}. In the following screenshots and the accompanying videos of complete reconfiguration sequences\footnote{\url{https://github.com/PetersTom/AmoebotVideos}} supply particles are colored green, demand roots red, and supply roots cyan.
The dark blue particles are part of the supply graph and therefore lie on a feather path from a demand root to a supply root.

Figure~\ref{fig:parallel} illustrates that our algorithm does indeed create a supply graph which is based on the shortest paths between supply and demand and hence facilitates parallel movement paths if the geometry allows. Activation sequences are randomized, so it is challenging to prove statements that capture which supply feeds which demand, but generally we observe that close supply and demand nodes will connect first.
An interesting open question in this context is illustrated in Figure~\ref{fig:bottleneck}: here we see two supplies and two demands, but the shortest paths have a common bottleneck, which slows down the reconfiguration in practice. Is there an effective way to include near-shortest paths in our supply graph to maximize parallelism? One could also consider temporarily adding particles to the symmetric difference. Note, though, that additional parallelism leads only to constant factor improvements in the running time; this is of course still meaningful in practice.

\begin{figure*}[t]
    \includegraphics[width=.48\textwidth]{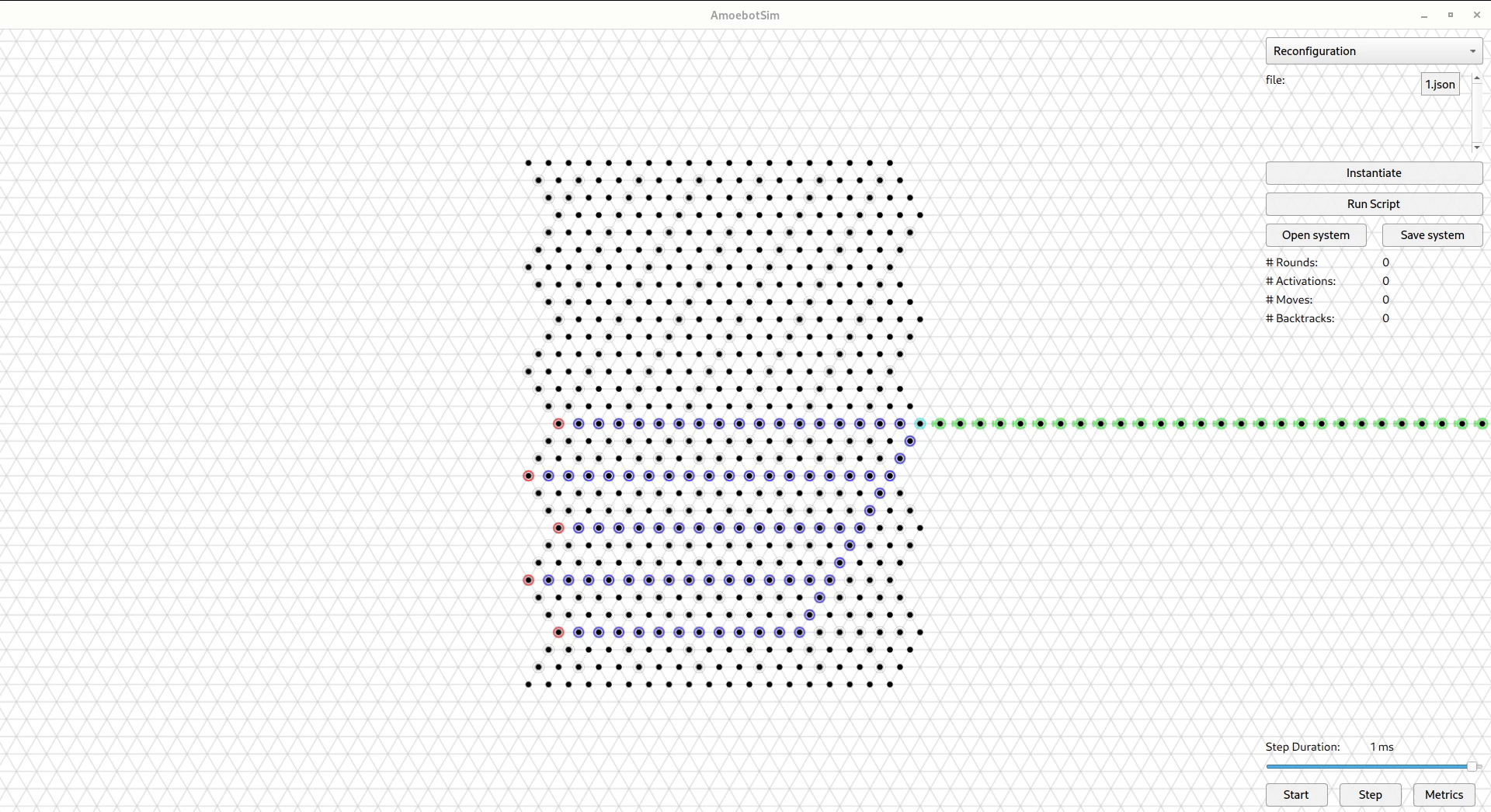}
    \hfill
    \includegraphics[width=.48\textwidth]{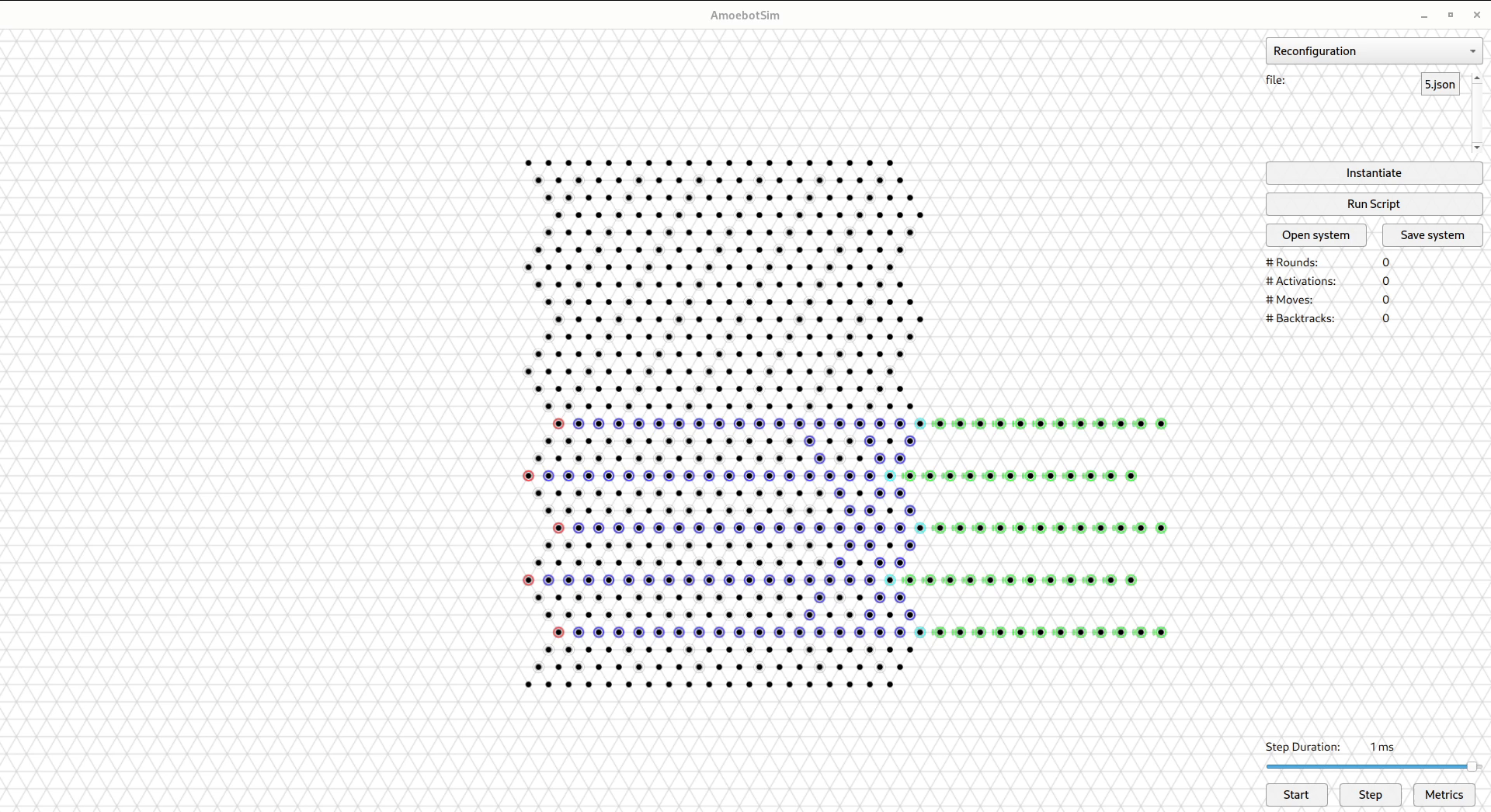}
    \caption{One supply and five demands vs. five supplies and five demands.}
    \label{fig:parallel}
\end{figure*}

\begin{figure}[t]
    \centering
    \includegraphics{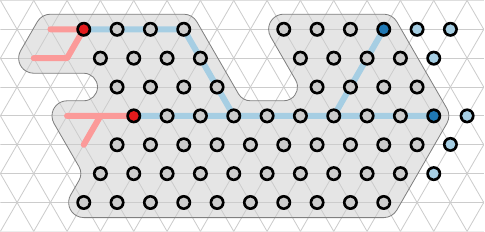}
    \caption{Shortest paths lead through a bottleneck, slowing down the algorithm in practice.}
    \label{fig:bottleneck}
\end{figure}

As we already hinted in the introduction, the running time of our algorithm is linked to the balance of supply in the feather trees. Recall that the trees are rooted at demand and grow towards supply; bubbles move from demand to supply. Because bubbles decide randomly at each junction which path to follow, each sub-tree will in expectation receive a similar amount of bubbles. If all sub-trees of a junction contain an equal amount of supply particles, then we say this junction is \emph{balanced}. If all junctions are balanced, then we say a feather tree is balanced; similarly we speak about a balanced supply graph.

The worst case running time of $O(n)$ rounds is triggered by unbalanced supply graphs. See, for example, Figure~\ref{fig:specialcases} (right), where at each junction the sub-tree rooted in the center carries all remaining supply, except for two particles. Such situations arise when there are many small patches of supply and only a few locations with demand. In the realistic scenario of shape repair we have exactly the opposite situation with many small damages (demand) and one large reservoir of supply stored for repairs. Here the supply graph will naturally be balanced and the running time is proportional to the diameter of the core $I\cup T$. Figure~\ref{fig:specialcases} (left) shows an example with only one supply and demand; here the running time of our algorithm is even proportional to the distance between supply and demand within $I\cup T$.

Feather trees are created to facilitate particles traveling on (crossing) shortest paths between supply and demand; they do not take the balance of the supply graph into account. A challenging open question in this context is whether it is possible to create supply graphs which retain the navigation properties afforded by feather tress but are at the same time balanced with respect to the supply.

\begin{figure}[b]
    \includegraphics[width=.48\textwidth]{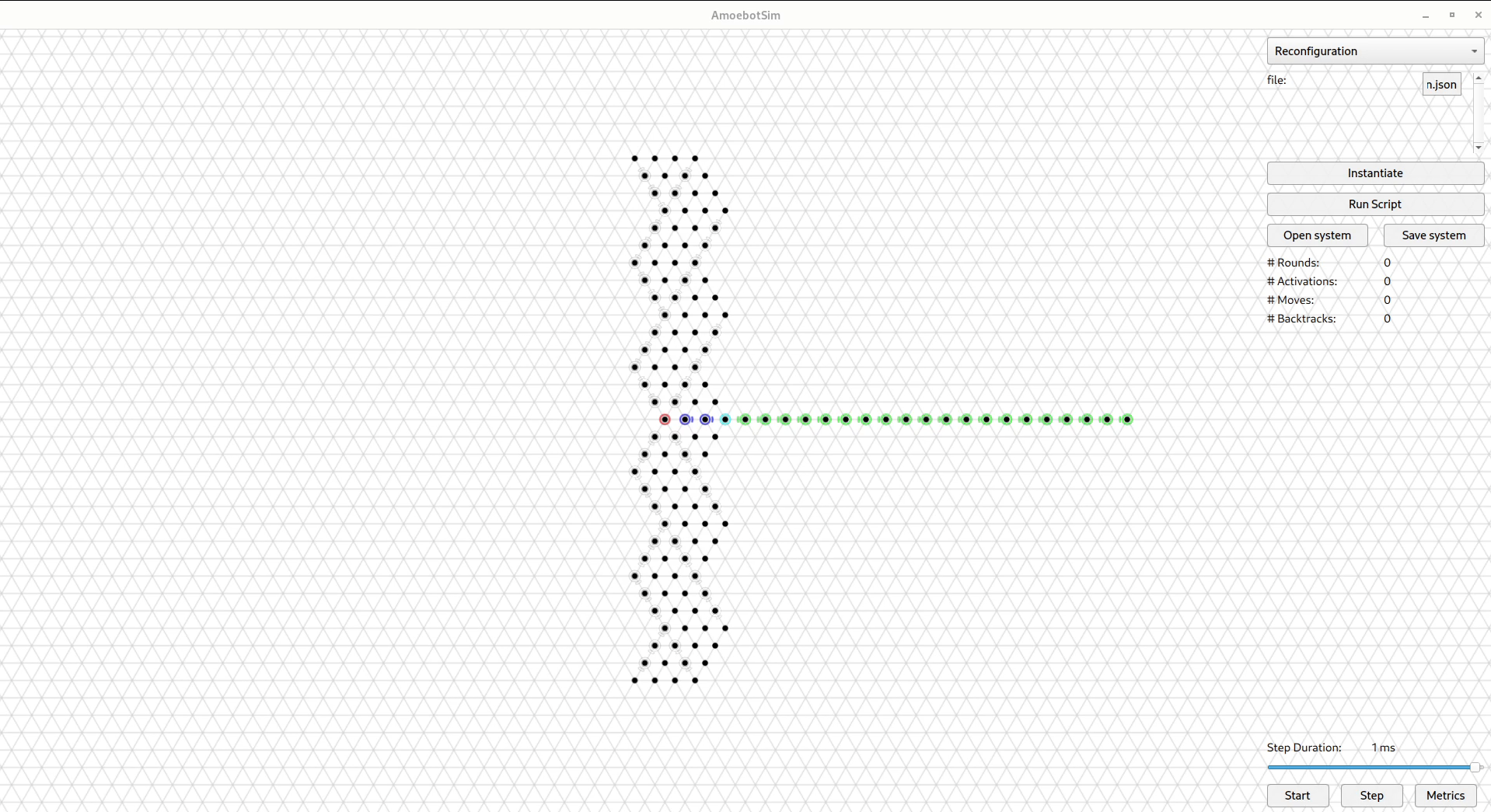}
    \hfill
    \includegraphics[width=.48\textwidth]{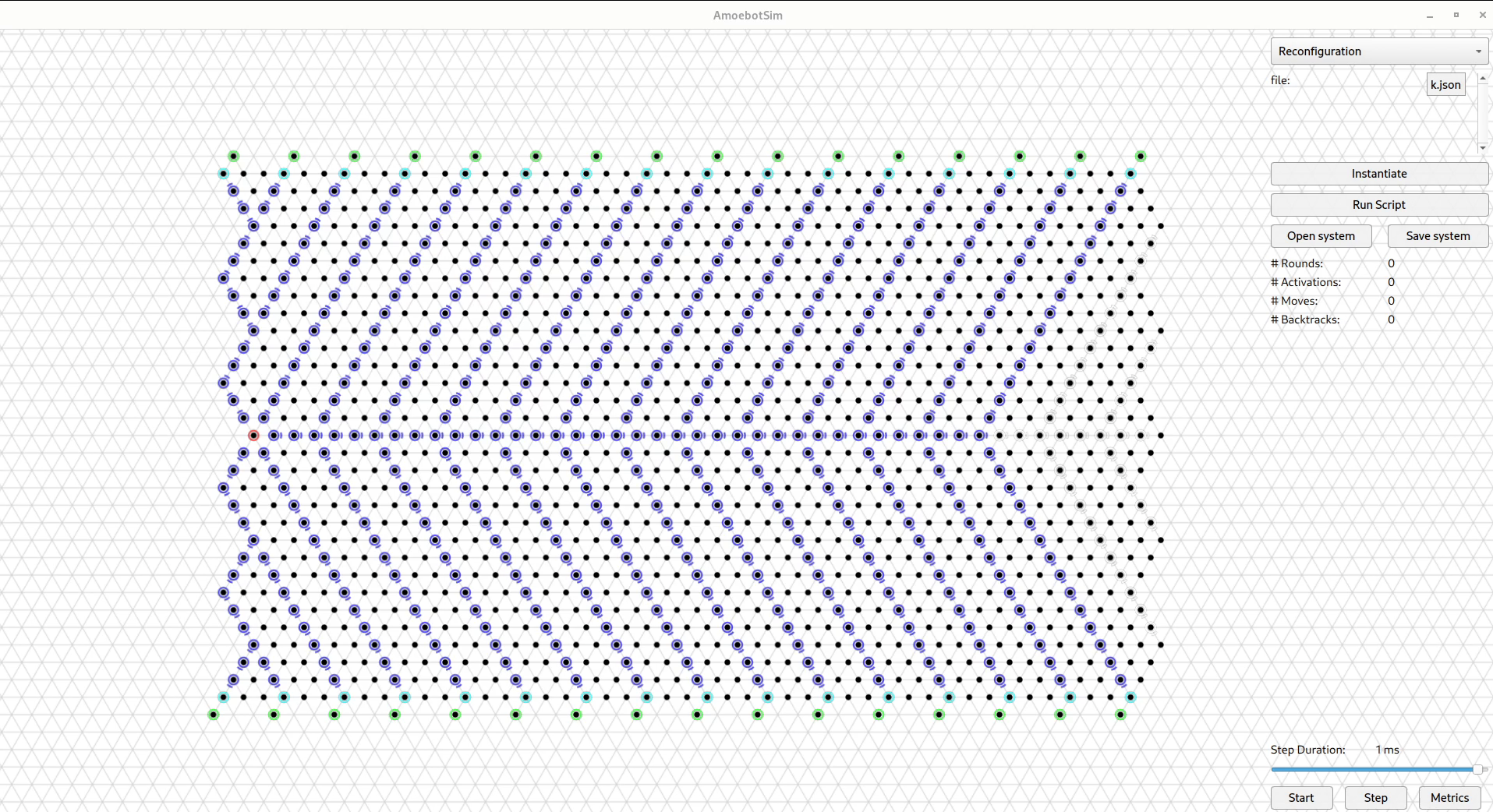}
    \caption{Left: Particles move directly between supply and demand. Right: Worst case configuration for back-tracking.}
    \label{fig:specialcases}
\end{figure}

\begin{figure*}[t]
    \centering
    \includegraphics[width=0.32\textwidth]{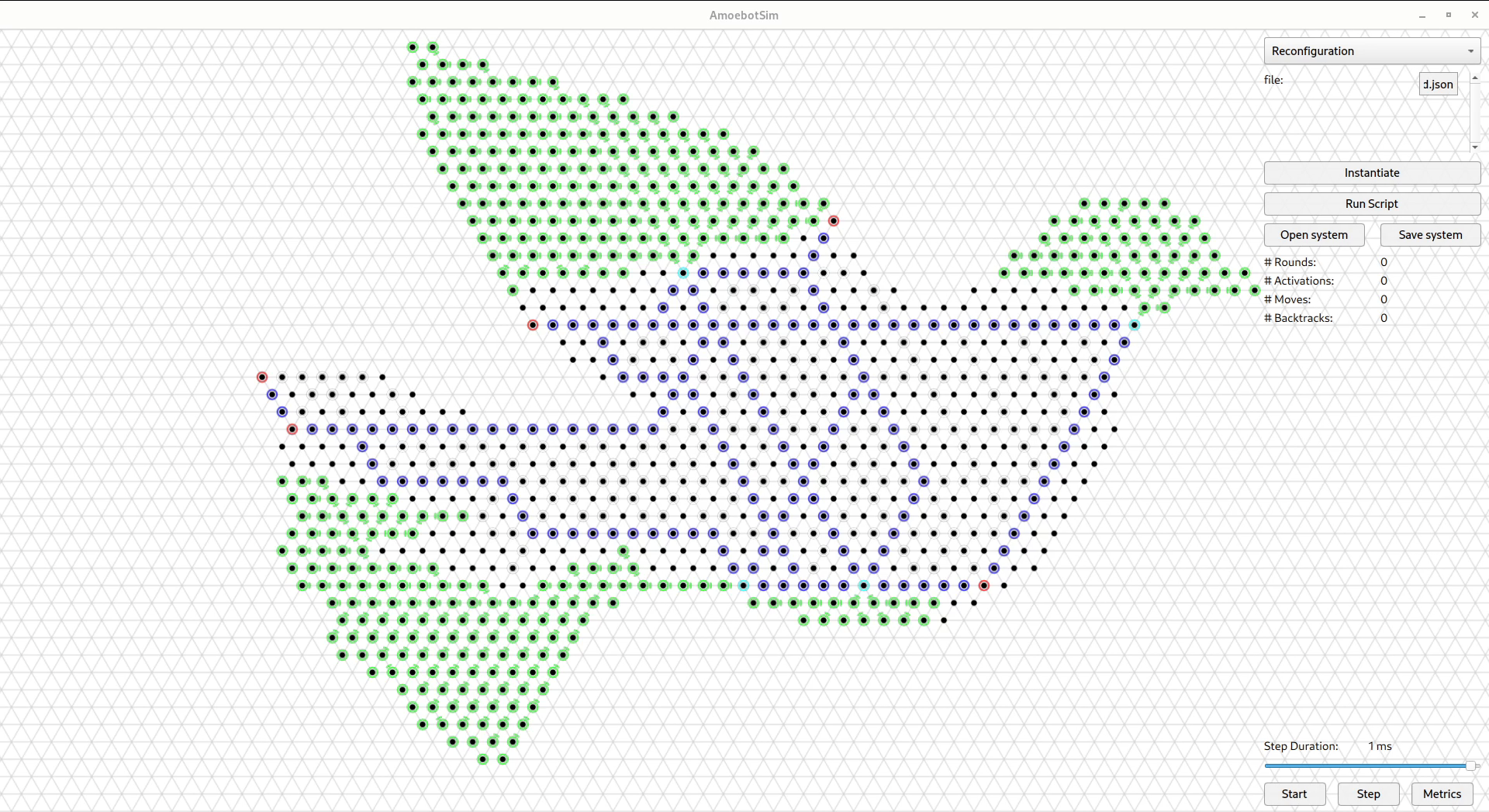}
    \hfill
    \includegraphics[width=0.32\textwidth]{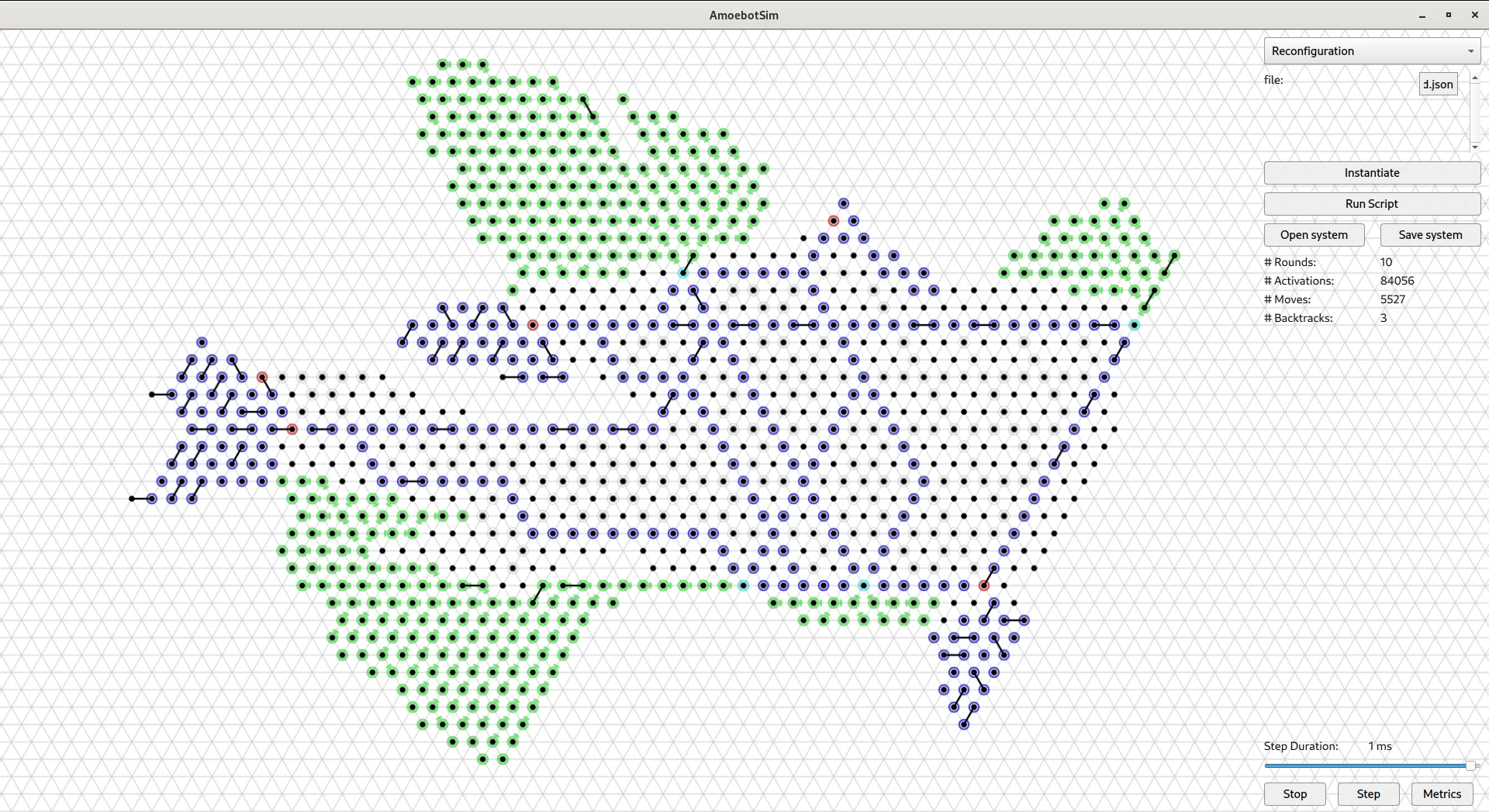}
    \hfill
    \includegraphics[width=0.32\textwidth]{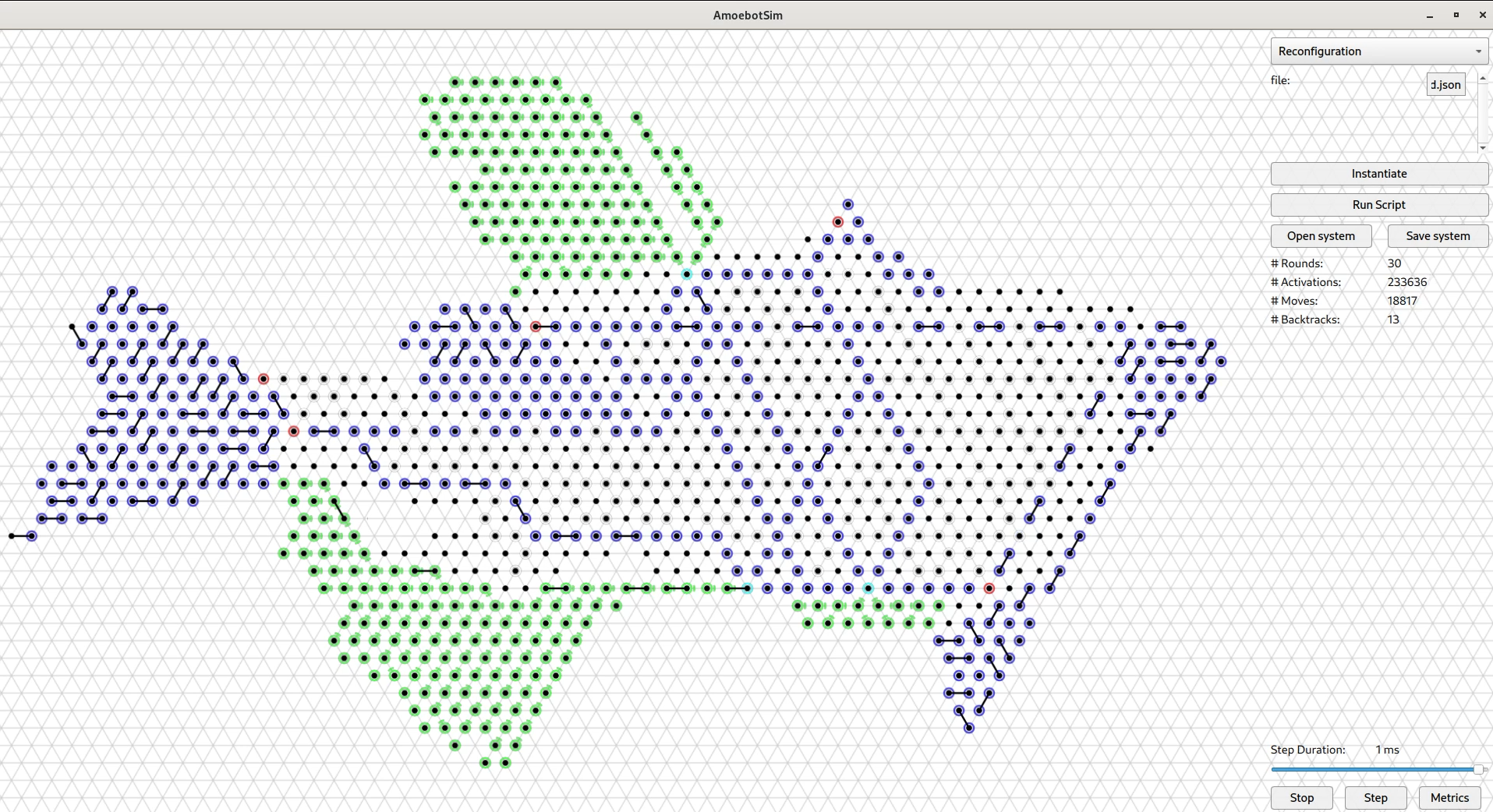}
    \vskip 2mm
    \includegraphics[width=0.32\textwidth]{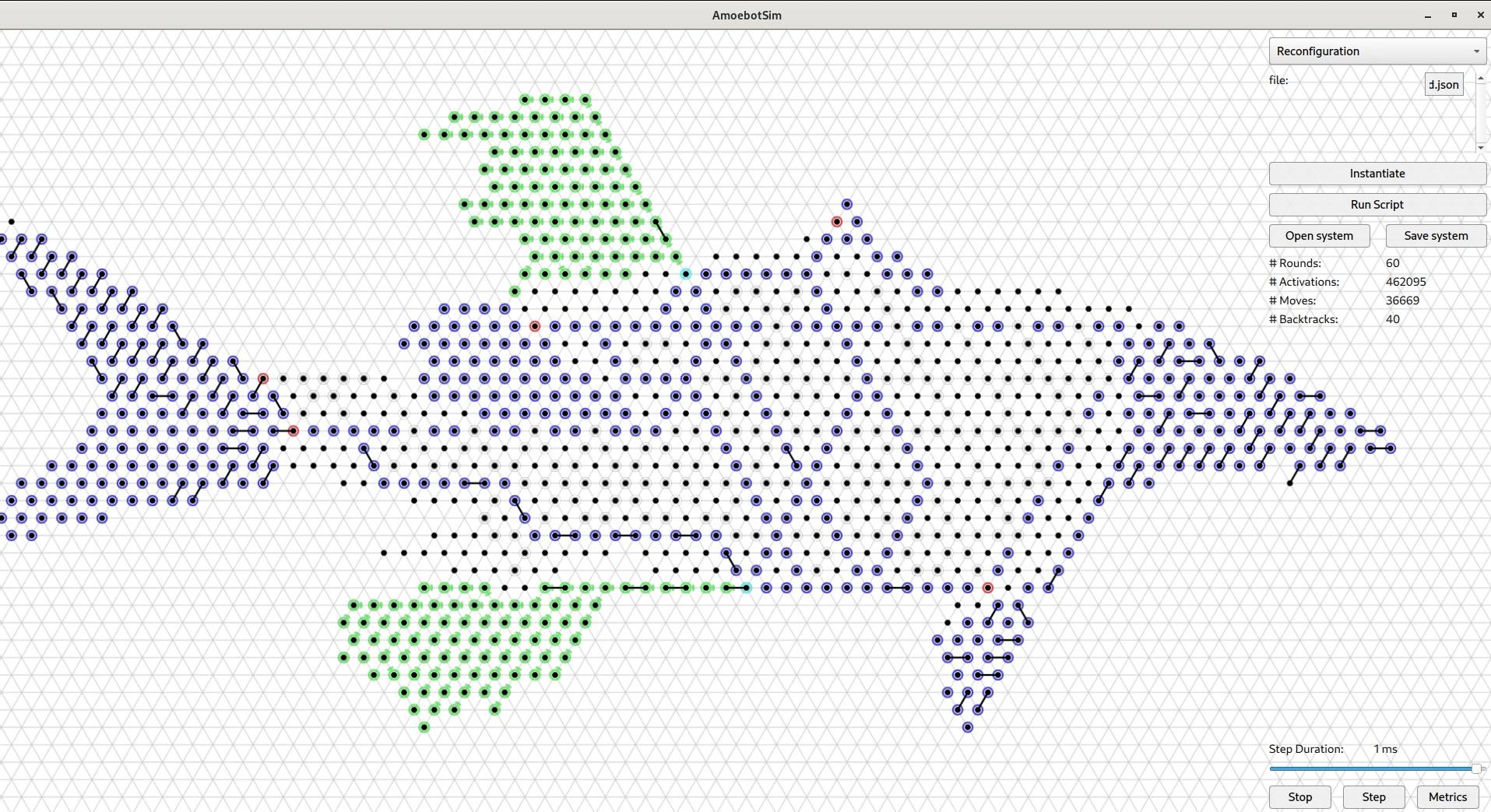}
    \hfill
    \includegraphics[width=0.32\textwidth]{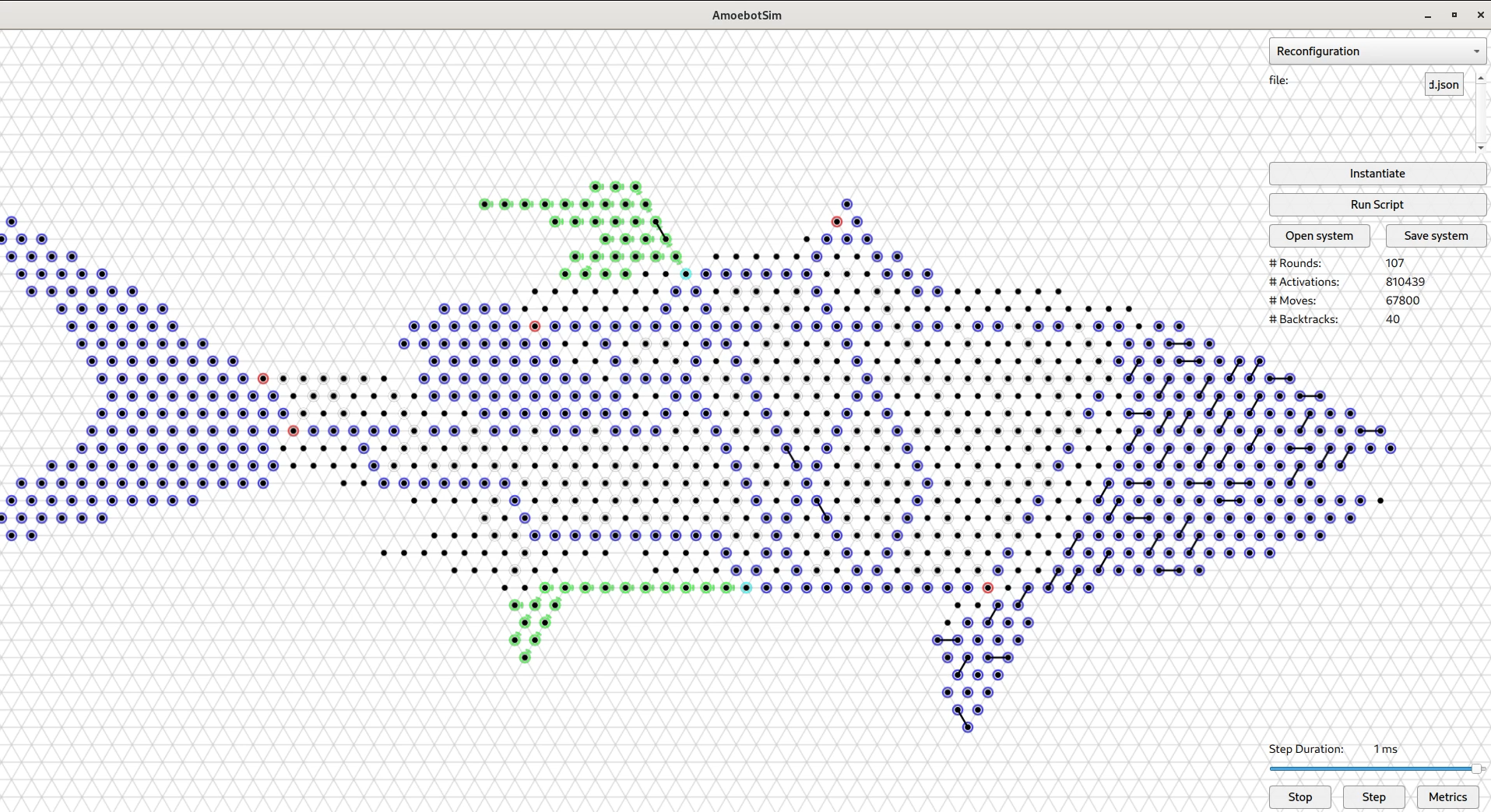}
    \hfill
    \includegraphics[width=0.32\textwidth]{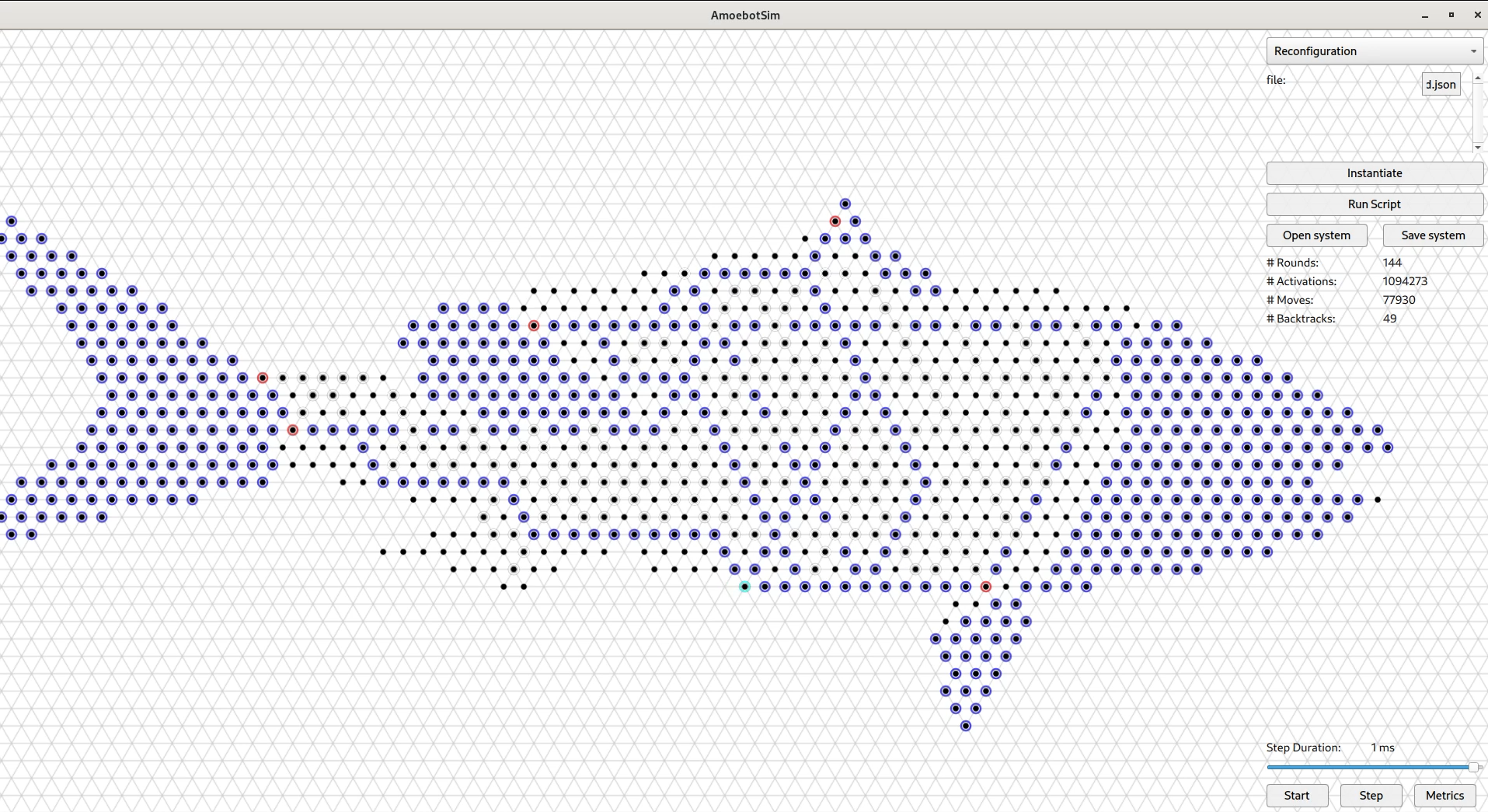}
    \caption{Dove to fish, full video at \url{https://github.com/PetersTom/AmoebotVideos}.}
    \label{fig:dovefish}
\end{figure*}

As a final example, Figure~\ref{fig:dovefish} gives an impression of the natural appearance of the reconfiguration sequences produced by our algorithm.

\backmatter

\bibliography{bibliography}

@article{Cheung2011,
	author = {Cheung, Kenneth C. and Demaine, Erik D. and Bachrach, Jonathan R. and Griffith, Saul},
	doi = {10.1109/TRO.2011.2132951},
	issn = {1552-3098},
	journal = {IEEE Transactions on Robotics},
	number = {4},
	pages = {718--729},
	title = {{Programmable Assembly With Universally Foldable Strings (Moteins)}},
	volume = {27},
	year = {2011}
}

@inproceedings{Cannon2016,
	author = {Cannon, Sarah and Daymude, Joshua J. and Randall, Dana and Richa, Andr{\'{e}}a W.},
	booktitle = {Proc. ACM Symposium on Principles of Distributed Computing (PODC)},
	doi = {10.1145/2933057.2933107},
	isbn = {9781450339643},
	pages = {279--288},
	title = {{A Markov Chain Algorithm for Compression in Self-Organizing Particle Systems}},
	year = {2016}
}

@article{DiLuna2019,
	author = {{Di Luna}, Giuseppe A. and Flocchini, Paola and Santoro, Nicola and Viglietta, Giovanni and Yamauchi, Yukiko},
	doi = {10.1007/s00446-019-00350-6},
	issn = {0178-2770},
	journal = {Distributed Computing},
	title = {{Shape formation by programmable particles}},
	year = {2020},
	pages = {69–101},
	volume = {33}
}

@InProceedings{Woods2013,
	Title = {Active self-assembly of algorithmic shapes and patterns in polylogarithmic time},
	Author = {Damien Woods and Ho{-}Lin Chen and Scott Goodfriend and Nadine Dabby and Erik Winfree and Peng Yin},
	Booktitle = {Proc. 4th Conference on Innovations in Theoretical Computer Science ({ITCS})},
	Year = {2013},
	pages = {353--354},
	doi = {10.1145/2422436.2422476}
}

@inproceedings{Gmyr2018,
	author = {Gmyr, Robert and Hinnenthal, Kristian and Kostitsyna, Irina and Kuhn, Fabian and Rudolph, Dorian and Scheideler, Christian and Strothmann, Thim},
	booktitle = {Proc. International Conference on DNA Computing and Molecular Programming (DNA)},
	doi = {10.1007/978-3-030-00030-1_8},
	pages = {122--138},
	title = {{Forming Tile Shapes with Simple Robots}},
	year = {2018}
}

@Article{Patitz2014,
	Title = {An introduction to tile-based self-assembly and a survey of recent results},
	Author = {Matthew J. Patitz},
	Journal = {Natural Computing},
	Year = {2014},
	Number = {2},
	Pages = {195--224},
	Volume = {13},
	Doi = {10.1007/s11047-013-9379-4}
}

@inproceedings{Derakhshandeh2014,
	author = {Derakhshandeh, Zahra and Dolev, Shlomi and Gmyr, Robert and Richa, Andr{\'{e}}a W. and Scheideler, Christian and Strothmann, Thim},
	booktitle = {Proc. 26th ACM Symposium on Parallelism in Algorithms and Architectures (SPAA)},
	doi = {10.1145/2612669.2612712},
	isbn = {9781450328210},
	pages = {220--222},
	title = {{Brief announcement: Amoebot---A New Model for Programmable Matter}},
	year = {2014}
}

@inproceedings{Derakhshandeh2015leader,
author = {Derakhshandeh, Zahra and Gmyr, Robert and Strothmann, Thim and Bazzi, Rida and Richa, Andr{\'{e}}a W. and Scheideler, Christian},
booktitle = {Proc. International Workshop on DNA-Based Computing (DNA)},
doi = {10.1007/978-3-319-21999-8_8},
pages = {117--132},
series = {LNCS 9211},
title = {{Leader Election and Shape Formation with Self-organizing Programmable Matter}},
year = {2015}
}

@InProceedings{Derakhshandeh2016,
	Title = {{Universal Shape Formation for Programmable Matter}},
	Author = {Zahra Derakhshandeh and Robert Gmyr and Andr{\'{e}}a W. Richa and Christian Scheideler and Thim Strothmann},
	Booktitle = {Proc. 28th Annual ACM Symposium on Parallelism in Algorithms and Architectures ({SPAA})},
	Year = {2016},
	Pages = {289--299},
	doi = {10.1145/2935764.2935784}
}

@inproceedings{Daymude2020,
author = {Daymude, Joshua J. and Gmyr, Robert and Hinnenthal, Kristian and Kostitsyna, Irina and Scheideler, Christian and Richa, Andr{\'{e}}a W.},
booktitle = {Proc. 21st International Conference on Distributed Computing and Networking},
doi = {10.1145/3369740.3372916},
isbn = {9781450377515},
pages = {1--10},
title = {{Convex Hull Formation for Programmable Matter}},
year = {2020}
}

@article{Mitchell1991,
	author = {Mitchell, Joseph S. B.},
	doi = {10.1007/BF01530888},
	issn = {1012-2443},
	journal = {Annals of Mathematics and Artificial Intelligence},
	number = {1},
	pages = {83--105},
	title = {{A new algorithm for shortest paths among obstacles in the plane}},
	volume = {3},
	year = {1991}
}

@techreport{Culberson1987,
author = {Culberson, Joseph C. and Reckhow, Robert A.},
institution = {University of Alberta},
group = {Department of Computing Science},
year = {1987},
title = {{Dent Diagrams: A Unified Approach to Polygon Covering Problems}},
note = {TR 87--14}
}

@inproceedings{dehkordi2014,
  title={{Increasing-Chord Graphs On Point Sets}},
  author={Dehkordi, Hooman R. and Frati, Fabrizio and Gudmundsson, Joachim},
  booktitle={Proc. International Symposium on Graph Drawing (GD)},
  pages={464--475},
  year={2014},
  series = {LNCS 8871},
  doi ={10.1007/978-3-662-45803-7_39}
}

@inproceedings{Porter2018,
author = {Porter, Alexandra and Richa, Andrea},
booktitle = {Proc. International Conference on Unconventional Computation and Natural Computation (UCNC)},
doi = {10.1007/978-3-319-92435-9_14},
pages = {188--203},
series = {LNCS 10867},
title = {{Collaborative Computation in Self-organizing Particle Systems}},
year = {2018}
}

@book{Ghosh2007,
place={Cambridge},
title={Visibility Algorithms in the Plane}, 
DOI={10.1017/CBO9780511543340}, 
publisher={Cambridge University Press}, 
author={Ghosh, Subir Kumar},
year={2007}
}

@article{Daymude2018,
archivePrefix = {arXiv},
arxivId = {1606.03642},
author = {Daymude, Joshua J. and Derakhshandeh, Zahra and Gmyr, Robert and Porter, Alexandra and Richa, Andr{\'{e}}a W. and Scheideler, Christian and Strothmann, Thim},
doi = {10.1007/s11047-017-9658-6},
eprint = {1606.03642},
issn = {15729796},
journal = {Natural Computing},
number = {1},
pages = {81--96},
title = {{On the Runtime of Universal Coating for Programmable Matter}},
volume = {17},
year = {2016}
}

@article{Geary2014,
author = {Geary, Cody and Rothemund, Paul W. K. and Andersen, Ebbe S.},
doi = {10.1126/science.1253920},
issn = {0036-8075},
journal = {Science},
number = {6198},
pages = {799--804},
title = {{A single-stranded architecture for cotranscriptional folding of RNA nanostructures}},
volume = {345},
year = {2014}
}

@inproceedings{Demaine2018b,
author = {Demaine, Erik D. and Hendricks, Jacob and Olsen, Meagan and Patitz, Matthew J. and Rogers, Trent A. and Schabanel, Nicolas and Seki, Shinnosuke and Thomas, Hadley},
booktitle = {DNA Computing and Molecular Programming},
doi = {10.1007/978-3-030-00030-1_2},
pages = {19--36},
title = {{Know When to Fold 'Em: Self-assembly of Shapes by Folding in Oritatami}},
year = {2018}
}

@article{Piranda2018,
author = {Piranda, Benoit and Bourgeois, Julien},
doi = {10.1007/s10514-018-9710-0},
issn = {0929-5593},
journal = {Autonomous Robots},
number = {8},
pages = {1619--1633},
title = {{Designing a quasi-spherical module for a huge modular robot to create programmable matter}},
volume = {42},
year = {2018}
}

@inproceedings{Naz2016,
author = {Naz, Andre and Piranda, Benoit and Bourgeois, Julien and Goldstein, Seth Copen},
booktitle = {Proc. 2016 IEEE 15th International Symposium on Network Computing and Applications (NCA)},
doi = {10.1109/NCA.2016.7778628},
isbn = {978-1-5090-3216-7},
pages = {254--263},
title = {{A distributed self-reconfiguration algorithm for cylindrical lattice-based modular robots}},
year = {2016}
}

@article{AndresArroyo2018,
author = {{Andr{\'{e}}s Arroyo}, Marta and Cannon, Sarah and Daymude, Joshua J. and Randall, Dana and Richa, Andr{\'{e}}a W.},
doi = {10.1007/s11047-018-9714-x},
issn = {1567-7818},
journal = {Natural Computing},
number = {4},
pages = {723--741},
title = {{A stochastic approach to shortcut bridging in programmable matter}},
volume = {17},
year = {2018}
}

@inproceedings{Dufoulon2021,
author = {Dufoulon, Fabien and Kutten, Shay and {Moses Jr.}, William K.},
booktitle = {Proc. 2021 ACM Symposium on Principles of Distributed Computing},
doi = {10.1145/3465084.3467900},
isbn = {9781450385480},
pages = {103--113},
title = {{Efficient Deterministic Leader Election for Programmable Matter}},
year = {2021}
}

@InProceedings{Cannon2019,
  author =	{Sarah Cannon and Joshua J. Daymude and Cem G{\"o}kmen and Dana Randall and Andr{\'e}a W. Richa},
  title =	{{A Local Stochastic Algorithm for Separation in Heterogeneous Self-Organizing Particle Systems}},
  booktitle =	{Approximation, Randomization, and Combinatorial Optimization. Algorithms and Techniques (APPROX/RANDOM 2019)},
  pages =	{54:1--54:22},
  series =	{Leibniz International Proceedings in Informatics (LIPIcs)},
  ISBN =	{978-3-95977-125-2},
  ISSN =	{1868-8969},
  year =	{2019},
  volume =	{145},
  doi =		{10.4230/LIPIcs.APPROX-RANDOM.2019.54}
}

@InProceedings{Daymude2021,
  author =	{Daymude, Joshua J. and Richa, Andr\'{e}a W. and Scheideler, Christian},
  title =	{{The Canonical Amoebot Model: Algorithms and Concurrency Control}},
  booktitle =	{35th International Symposium on Distributed Computing (DISC)},
  pages =	{20:1--20:19},
  series =	{Leibniz International Proceedings in Informatics (LIPIcs)},
  ISBN =	{978-3-95977-210-5},
  ISSN =	{1868-8969},
  year =	{2021},
  volume =	{209},
  doi =		{10.4230/LIPIcs.DISC.2021.20}
}

@misc{Daymude2021a,
arxivId = {2111.09449},
author = {Daymude, Joshua J. and Richa, Andr{\'{e}}a W. and Scheideler, Christian},
journal = {arxiv},
month = {nov},
title = {{Local Mutual Exclusion for Dynamic, Anonymous, Bounded Memory Message Passing Systems}},
url = {http://arxiv.org/abs/2111.09449},
year = {2021}
}

@inproceedings{disc22,
  title={Brief Announcement: An Effective Geometric Communication Structure for Programmable Matter},
  author={Kostitsyna, Irina and Peters, Tom and Speckmann, Bettina},
  booktitle={36th International Symposium on Distributed Computing (DISC 2022)},
  year={2022},
  organization={Schloss Dagstuhl-Leibniz-Zentrum f{\"u}r Informatik}
}

@article{boulinier2008space,
  title={Space efficient and time optimal distributed BFS tree construction},
  author={Boulinier, Christian and Datta, Ajoy K. and Larmore, Lawrence L. and Petit, Franck},
  journal={Information processing letters},
  volume={108},
  number={5},
  pages={273--278},
  year={2008},
  publisher={Elsevier}
}

\end{document}